%% file: rw.tex
\documentclass{vldb}
\usepackage{amssymb}
\usepackage{amsmath}
\usepackage{color}
\usepackage{times}
\usepackage{cite}
\newtheorem{theorem}{Theorem}

\newtheorem{definition}{Definition}
\usepackage{algorithm}
\usepackage{algorithmic}
\usepackage{hyperref}
\usepackage{caption,subcaption} %for figures
\usepackage{enumitem} %for reducing space between items

\usepackage{mycomments}
\def\commenton{1}  %% Comment out this line to hide comments 
\title{Walk, Not Wait: Faster Sampling Over Online Social Networks}
\numberofauthors{1}
\author
{
\alignauthor
Azade Nazi$^{\ddag}$,
Zhuojie Zhou$^{\dag\dag}$,
Saravanan Thirumuruganathan$^{\ddag}$,
Nan Zhang$^{\dag\dag}$, 
Gautam Das$^{\ddag}$
\\
\affaddr {
$^{\ddag}$University of Texas at Arlington;
$^{\dag\dag}$George Washington University
}
{\email
{
$^{\ddag}$\small{\{azade.nazi@mavs,~saravanan.thirumuruganathan@mavs,~gdas@cse\}.uta.edu},
$^{\dag\dag}$\small{\{rexzhou,nzhang10\}@gwu.edu}
}
}
}
\date{}

\makeatletter
\def\@copyrightspace{\relax} %removes copyright box
\def\@maketitle{\newpage
 \null
 \setbox\@acmtitlebox\vbox{%
\baselineskip 20pt
\vskip 2em                   % Vertical space above title.
   \begin{center}
    {\ttlfnt \@title\par}       % Title set in 18pt Helvetica (Arial) bold size.
    \vskip 1.5em                % Vertical space after title.
{\subttlfnt \the\subtitletext\par}\vskip 1.25em%\fi
    {\baselineskip 16pt\aufnt   % each author set in \12 pt Arial, in a
     \lineskip .5em             % tabular environment
     \begin{tabular}[t]{c}\@author
     \end{tabular}\par}
    \vskip 1.5em               % Vertical space after author.
   \end{center}}
 \dimen0=\ht\@acmtitlebox
 \unvbox\@acmtitlebox
 \ifdim\dimen0<0.0pt\relax\vskip-\dimen0\fi}
\makeatother
\begin{document}
\maketitle

\abstract{
In this paper, we introduce a novel, general purpose, technique for faster sampling of nodes over an online social network. Specifically, unlike traditional random walk which wait for the convergence of sampling distribution to a predetermined target distribution - a waiting process that incurs a high query cost - we develop WALK-ESTIMATE, which starts with a much shorter random walk, and then proactively estimate the sampling probability for the node taken before using acceptance-rejection sampling to adjust the sampling probability to the predetermined target distribution. We present a novel backward random walk technique which provides provably unbiased estimations for the sampling probability, and demonstrate the superiority of WALK-ESTIMATE over traditional random walks through theoretical analysis and extensive experiments over real world online social networks.}

\input{1-intro}
\input{2-preliminaries}
\input{4-walk}
\input{5-estimate}

\input{6-discussion}
\input{7-exp}

\input{8-related_work}

\section{Final Remarks}

In this paper, we developed WALK-ESTIMATE, a general purpose technique for faster sampling of nodes over an online social network with any target (sampling) distribution. Our key idea is to conduct a random walk for a small number of steps, and follow it with a proactive estimation of the sampling distribution of the node encountered before applying acceptance-rejection sampling to achieve the target distribution. Specifically, we presented two main components of WALK-ESTIMATE, WALK which determines the number of steps to walk, and ESTIMATE which enables an unbiased estimation of the sampling distribution. Theoretical analysis and extensive experimental evaluations over synthetic graphs and real-world online social networks demonstrated the superiority of our technique over the existing random walks.

\begin{figure}[ht]
	\centering
	\hspace{-0.4in}
	\subcaptionbox{Relative Error vs Query Cost\label{fig:gplus_srw_relE_vs_QueryCost_Synthetic}}{
		\vspace{-0.2in}
		\includegraphics[scale=0.25]{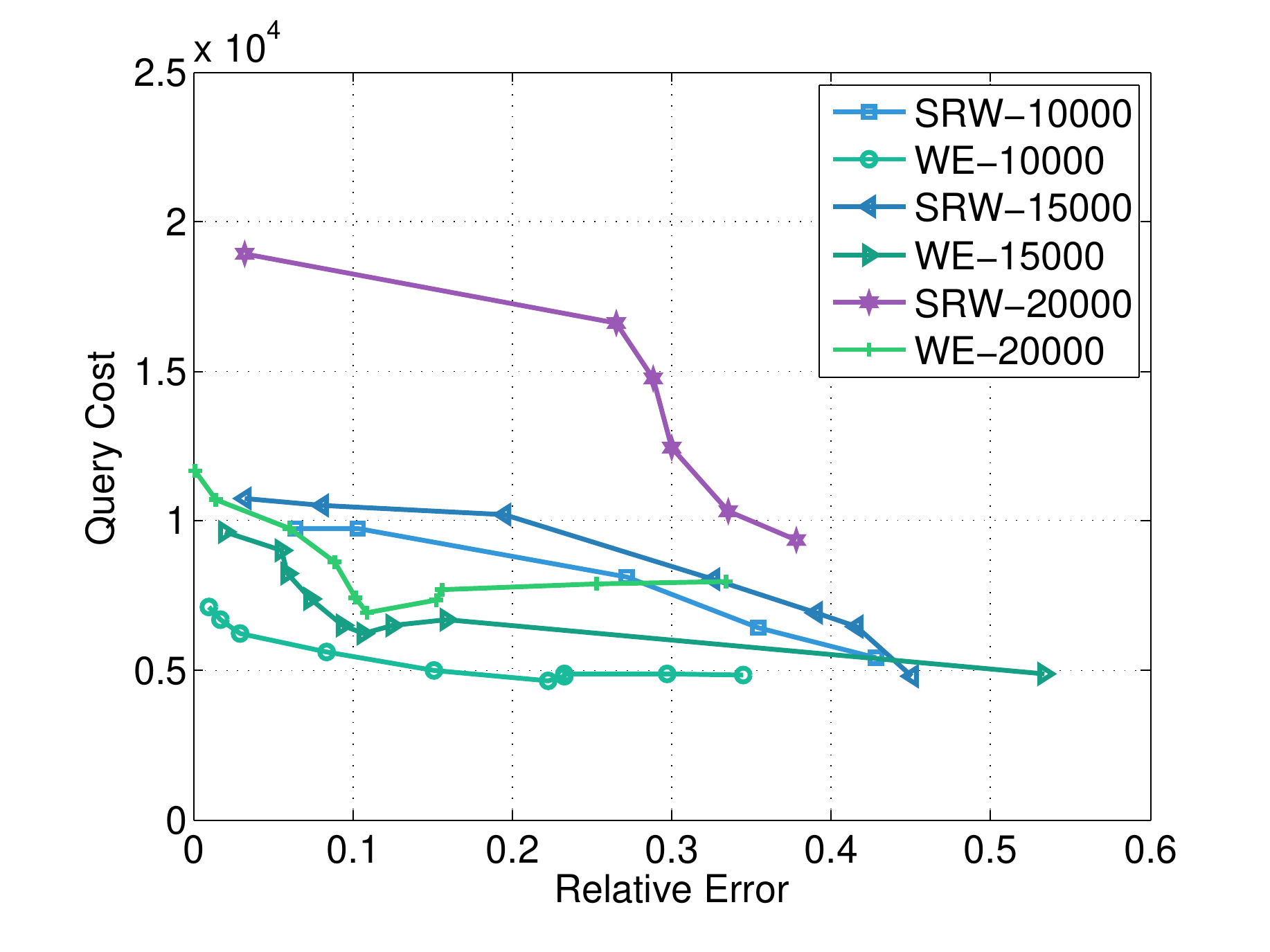}
		\hspace{-0.4in}
	}
	\quad 
	\subcaptionbox{Relative Error vs Number of Samples \label{fig:gplus_srw_relE_vs_NumSamples_Synthetic}}{
		\vspace{-0.2in}
		\includegraphics[scale=0.25]{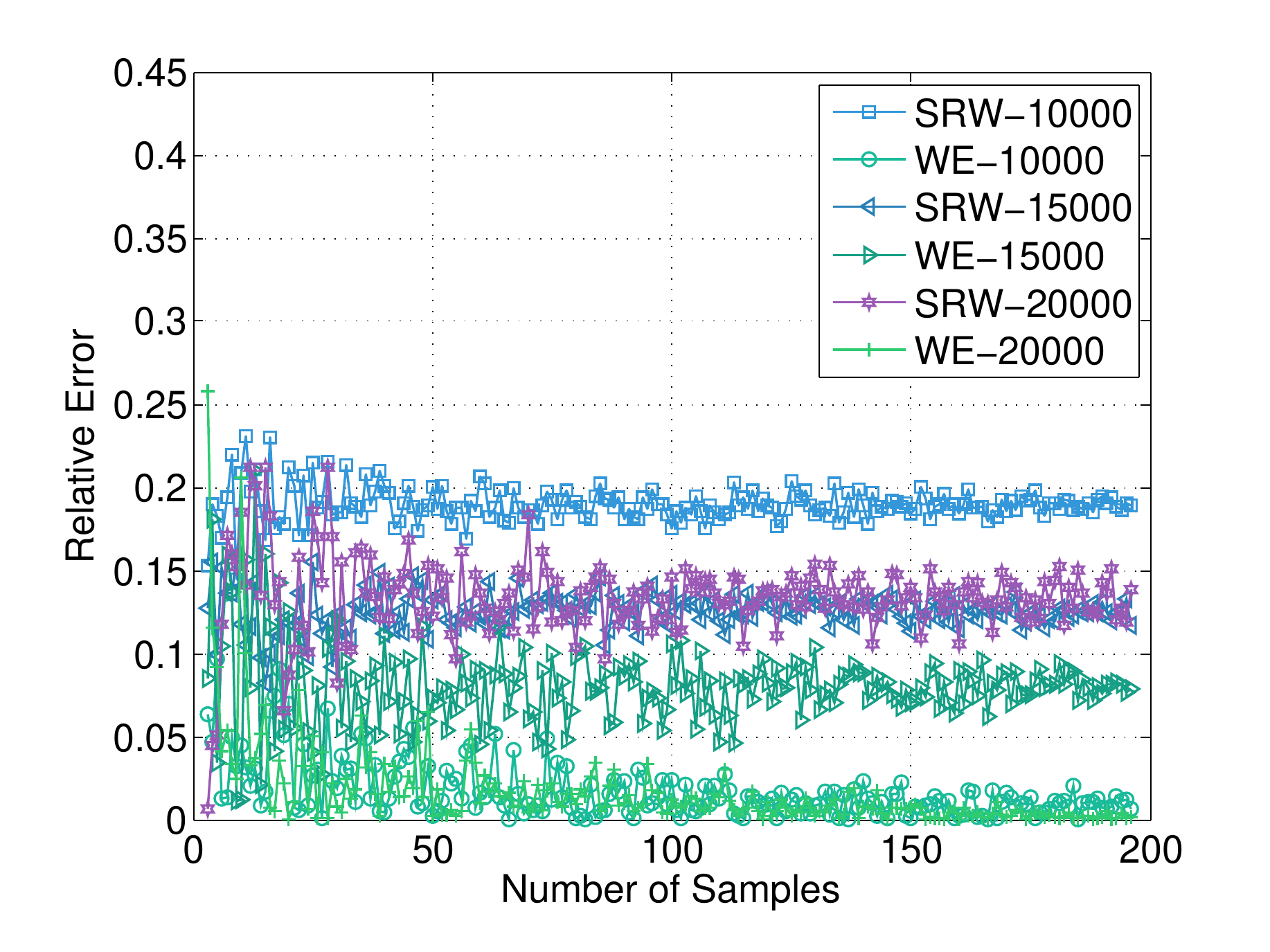}
		\hspace{-0.5in}
	}
	\vspace{-0.1in}
	\caption{Average Degree Estimation in Synthetic graphs (SRW).}\label{fig:syntheticRes}
	%\vspace{-0.2in}
\end{figure}

\begin{figure}[ht]
 \centering
	\hspace{-0.4in}
        \subcaptionbox{Scale-Free Graph: Comparison of PDF \label{fig:rev2BAPDF}}{
                \vspace{-0.2in}
                \includegraphics[scale=0.26]{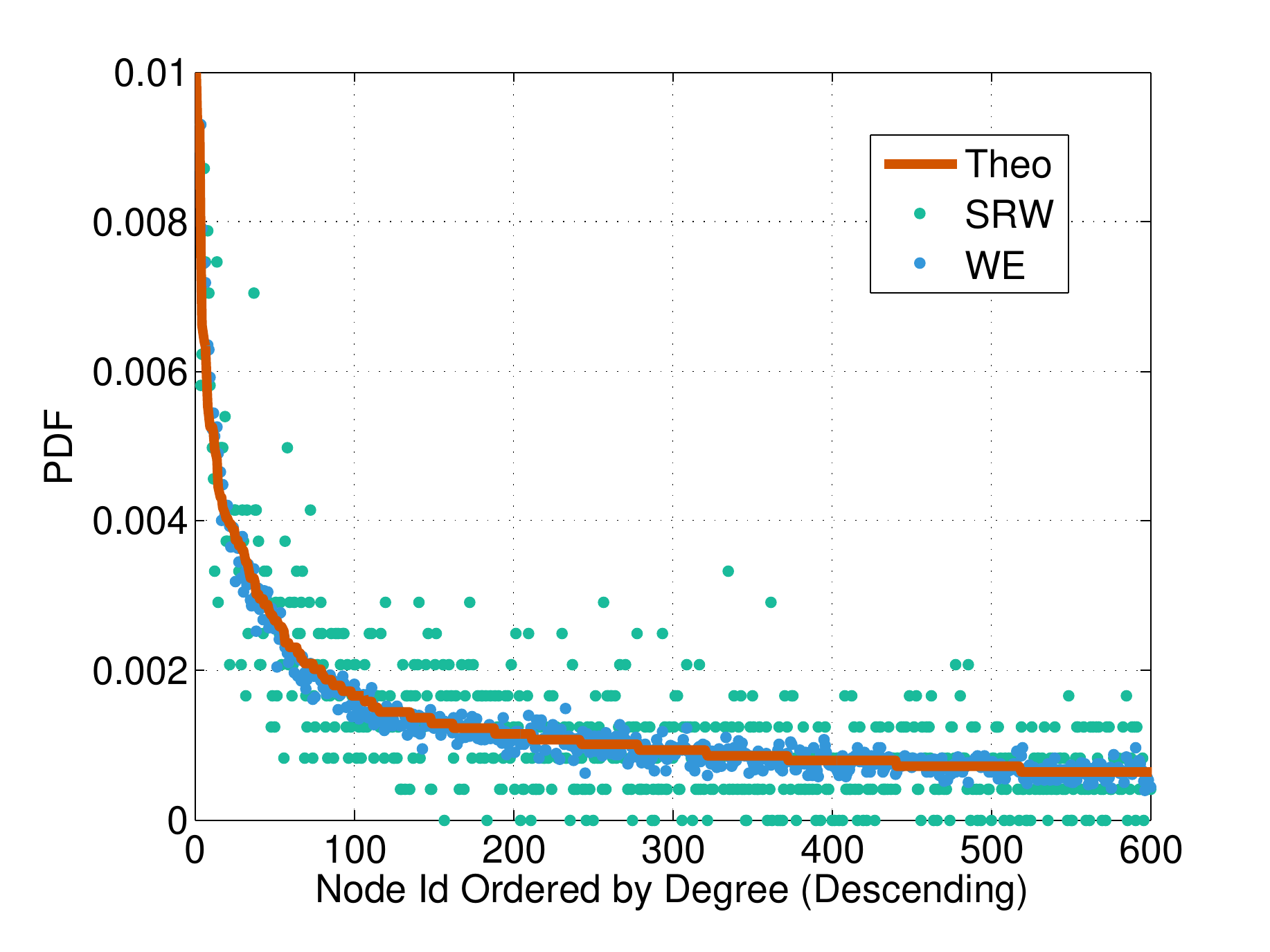}
                \hspace{-0.4in}
        }
        \quad 
        \subcaptionbox{Scale-Free Graph : Comparison of CDF \label{fig:rev2BACDF}}{
                \vspace{-0.2in}
                \includegraphics[scale=0.26]{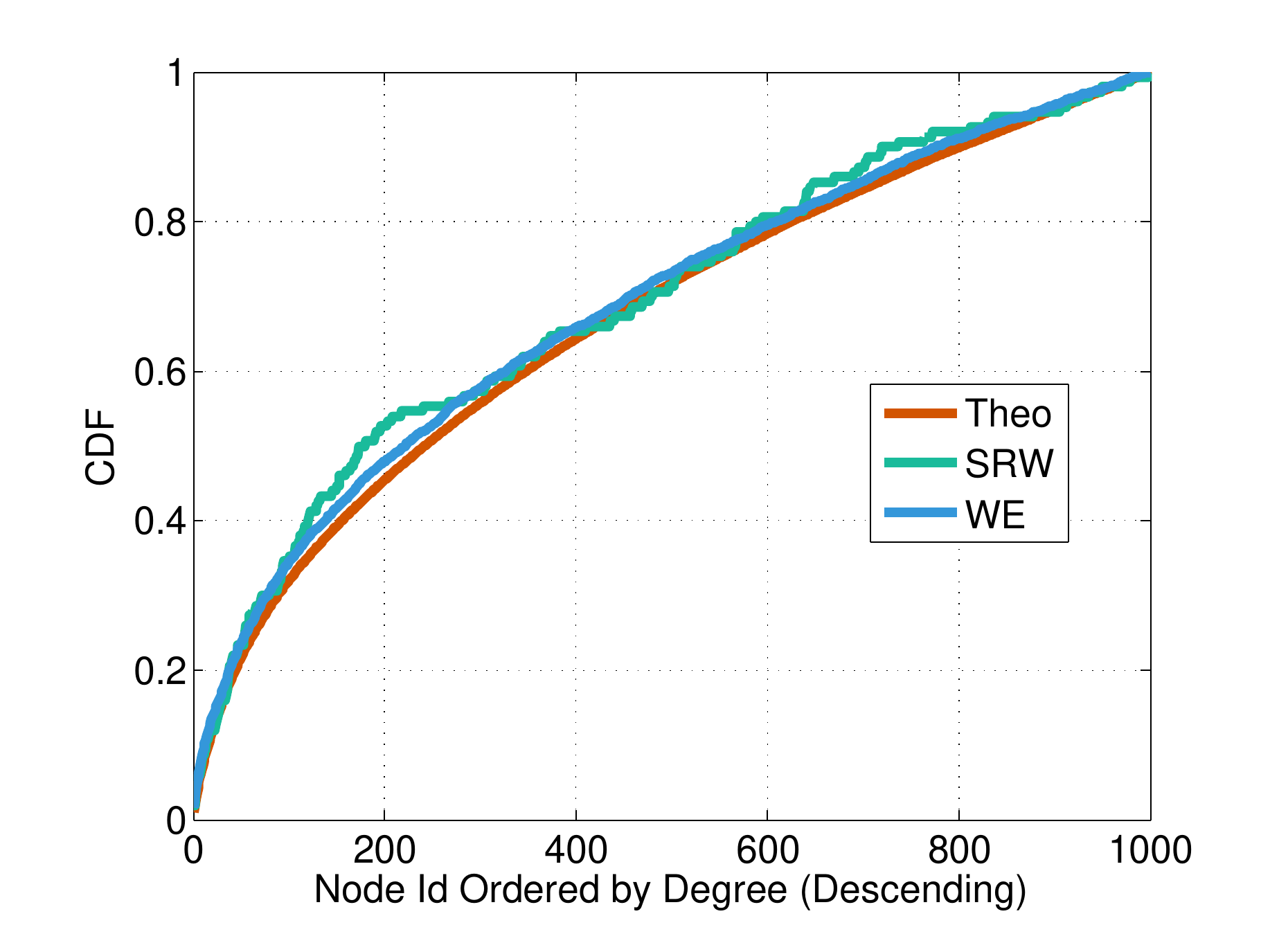}
                \hspace{-0.4in}
        }
        %\vspace{-0.1in}
        \caption{Comparing Sampling Distribution for Scale-Free Graph}
        \label{fig:rev2CompareBiasBA}
\end{figure}

\bibliographystyle{abbrv}
\bibliography{rw}  
\end{document}

%% file: 1-intro.tex
\section{Introduction}

Online social networks often feature a web interface that only allows local-neighborhood queries - i.e., given a user of the online social network as input, the system returns the immediate neighbors of the user. In this paper, we address the problem of enabling third-party analytics over an online social network through its restrictive web interface. As demonstrated by a wide range of existing services (e.g., Twitris, Toretter, AIDR), such third-party analytics applications benefit not only social network users, social scientists, but the entire society at large (e.g., through epidemic control).

\subsection{Problem of Existing Work}

The restrictive local-neighborhood-only access interface makes it extremely difficult for a third party to crawl all data from an online social network, as a complete crawl requires as many queries as the number of users in the social network. To address this challenge and enable analytics tasks such as aggregate estimation through the restrictive access interface, many existing studies resort to the {\em sampling} of users from the online social network. If we consider a social network as a graph, the idea here is to first draw a sample of nodes from the graph, and then generate (statistically accurate) aggregate estimations based on the sample nodes.

The nature of the interface limitation - i.e., allowing only local-neighborhood queries - makes {\em random walk} based Monte Carlo Markov Chain (MCMC) methods an ideal fit for the sampling of users from an online social network. Intuitively, a random walk starts from an arbitrary user, and then randomly moves to one of its neighbors selected randomly according to a pre-determined probability distribution (namely the {\em transition probability}). The movement continues for a number of steps, namely the ``{\em burn-in period}'' before the node being selected is taken as a sample.

A critical problem with the existing random walk techniques, however, is the long burn-in period it requires and, therefore, the significant query cost it incurs for drawing a sample. Since many online social networks limit the number of queries one can issue (e.g., from an IP address) within a period of time (e.g., Twitter allows , every 15 minutes, only 15 API requests to retrieve ids of a user's followers) the high query cost limits the sample size one can draw from the social network and, consequently, the accuracy of analytics tasks.

To understand why the burn-in period is required, an important observation is that, before a sample can be used for analytical purposes, we must know the {\em sampling distribution} - i.e., the probability for each node in the graph to be sampled - because without such knowledge, one might ``over-consider'' certain parts of the graph in the analytics tasks, leading to errors such as biased aggregate estimations. However, since a third party has no knowledge of even the global graph topology, it seems infeasible to compute the sampling distribution for a random walk. Fortunately, the property of MCMC methods ensures that, as a random walk grows longer, the sampling distribution becomes asymptotically close to a {\em stationary distribution} that can be computed from the design of transition probabilities alone. For example, with a simple random walk (featuring a uniform transition distribution - see Section~\ref{subsec:traditionalRW} for details), the stationary probability for a node to be sampled is always proportional to its degree, no matter how the global graph topology looks like.

It is also the availability of such a stationary distribution that leads to the mandate of a burn-in period. Note that, while the MCMC property ensures asymptotic convergence to the stationary distribution, the actual convergence process can be slow - and the length of burn-in required is essentially uncomputable without knowledge of the entire graph topology \cite{Levin:2008,Mohaisena}. Facing this problem, what the existing techniques can do is to either set a conservatively large burn-in period \cite{mislove2007measurement,Mohaisena}, or use one of the heuristic convergence monitors and ``wait'' for the sampling distribution to converge to its stationary value. In either case, the sampling process may require a large number of queries during the burn-in period.

\subsection{Our Idea: WALK-ESTIMATE}

In this paper, our objective is to significantly reduce the query cost of node sampling over an online social network by (nearly) eliminating the costly waiting process. Of course, as one can see from the above discussions, if we do not wait for the convergence to stationary distribution, we must somehow estimate the probability for our short walk to take a node as a sample (i.e., the node's {\em sampling probability}) before we can use the node as a sample. This is exactly what we do - i.e., we introduce a novel idea of having a (much) shorter, say $t$-step, walk, before taking a node $v$ as a sample candidate, but follow it up with a {\em proactive} process which estimates $v$'s sampling probability $p_t(v)$ - i.e., the probability for our walk to reach node $v$ at Step $t$, so that we can then use acceptance-rejection sampling to ``correct'' the sampling probability to the desired distribution. As we shall prove in the paper, even though the acceptance-rejection step introduces additional query cost, the savings from having a shorter walk in the first place far outweighs the additional cost, leading to a significantly more efficient sampling process.

Based on this idea, we develop Algorithm WALK-ESTIMATE. The algorithm takes as input a random walk based MCMC sampler, and produces samples according to the exact same distribution as the input sampler - i.e., the stationary distribution of the MCMC process. One can see that this design makes WALK-ESTIMATE transparent to the desired target distribution, making it a swap-in replacement for any random walk sampler being used (e.g., simple random walk \cite{Lovasz1993,Levin:2008}, Metropolis-Hastings random walk \cite{Lovasz1993,Levin:2008}.  As we shall demonstrate through theoretical analysis and experimental evaluation over real-world social networks, while the proactive probability-estimation process may consume a small number of queries, the significant savings from the shorter walk more than offset the additional consumption, and lead to a much more efficient sampling process over online social networks.

\subsection{Outline of Technical Results}

We now provide an overview of the main technical results in the design of WALK-ESTIMATE. The algorithm is enabled by two main components: WALK and ESTIMATE. The WALK component determines how many, say $t$, steps to (randomly) transit before taking a node $v$ as a candidate (for sampling), and then calls upon the ESTIMATE component for an estimation of the probability for the walk to reach $v$ after $t$ steps. Based on the estimated probability, WALK then performs acceptance-rejection sampling on $v$ to determine if it should be included in the final (output) sample.

For the WALK component, we start by developing IDEAL-WALK, an impractical sampler which makes two ideal-case assumptions: One is access to an oracle that precisely compute the $p_t(v)$, i.e., the probability for the walk to reach a node $v$ at Step $t$. The other is access to certain {\em global} topological information about the underlying graph - e.g., $|E|$, the total number of edges in the graph, $\mathcal{D}(G)$, the graph diameter, $\lambda$, the spectral gap of the transition matrix, and $d_{\max}$, the maximum degree of a node in the graph - so that IDEAL-WALK can determine the optimal number of steps $t$ to walk. We rigidly prove that no matter what the target distribution is (barring certain extreme cases, e.g., when the distribution is 1 on the starting node and 0 on all others), IDEAL-WALK {\em always} outperforms its corresponding traditional random walk algorithm. Further, it also produce samples with absolutely zero bias (while random walks often cannot, depending on the graph topology). We also demonstrate through analysis of multiple theoretical graph models the significance of such efficiency enhancements for the sampling process.

Of course, IDEAL-WALK makes two unrealistic assumptions, which we remove through the design of Algorithms WALK and ESTIMATE, respectively. Algorithm WALK removes the assumption of access to global parameters by requiring access to only one parameters (besides the local neighborhood of the current node): $\bar{\mathcal{D}}(G)$, i.e., an upper bound on the diameter of the graph - which is often easy to obtain (e.g., it is commonly believed that 8 to 10 is a safe bet for real-world online social networks \cite{backstrom2012four,mislove2007measurement}). 

Algorithm ESTIMATE, on the other hand, estimates $p_t(v)$, i.e., the probability for Algorithm WALK to sample node $v$ at Step $t$. To illustrate our main idea here, we start by developing UNBIASED-ESTIMATE, a simple algorithm which takes a {\em backward} random walk from Node $v$ for estimating $p_t(v)$. We rigidly prove the unbiasedness of the estimation produced by UNBIASED-ESTIMATE. Nonetheless, we also note its problem: a high estimation variance which grows rapidly with the number of backward steps one has to take for producing the estimation. Since the error of estimation is determined by both bias and variance, the high variance produced by UNBIASED-ESTIMATE introduces significant error in the estimation of $p_t(v)$.

To address the problem of UNBIASED-ESTIMATE, we introduce two main ideas for variance reduction in developing Algorithm ESTIMATE, our final algorithm for estimating $p_t(u)$: One is {\em initial crawling}, i.e., the crawl of the $h$-hop neighborhood (where $h$ is a small number like 2 or 3) of the starting node to reduce the number of backward steps required by Algorithm ESTIMATE, and the second is {\em weighted sampling}, i.e., to carefully design the transition matrix of the backward random walk process to reduce the variance of estimation. We shall demonstrate through experimental evaluation that Algorithm ESTIMATE significantly reduces the estimation variance for $p_t(u)$. Finally, we combine Algorithms WALK and ESTIMATE to produce Algorithm WALK-ESTIMATE.

\subsection{Summary of Contributions}

We make the following main contributions in this paper:

\begin{itemize}[noitemsep,nolistsep]
\item We propose a novel idea of WALK-ESTIMATE, a swap-in replacement for any random walk sampler which forgoes the long burn-in period and instead uses a proactive sampling probability estimation step to produce samples of a desired target distribution.
\item To demonstrate the superiority of our WALK step - i.e., performing a short random walk followed by acceptance-rejection sampling to reach the target distribution - we rigidly prove that, given a reasonable sample-bias requirement, no matter what the graph topology or target distribution is (barring certain extreme cases, e.g., when the distribution is 1 on the starting node and 0 on all others), a short random walk followed by acceptance-rejection sampling always outperforms its corresponding traditional random walk process.
\item For the ESTIMATE step, we introduce a novel UNBIASED-ESTIMATE algorithm which uses a small number of queries to produce a provably unbiased estimation of the sampling probability of a given node. In addition, we also propose two heuristics, initial crawling and weighted sampling, to reduce the variance (and consequently error) of an estimation.
\item Our contributions also include extensive experiments over real-world online social networks such as Google Plus, which confirm the significant improvement offer by our WALK-ESTIMATE algorithm over traditional random walks such as simple random walk and Metropolis-Hastings random walk.
\end{itemize}

\subsection{Paper Organization}

The rest of the paper is organized as follows. We discuss preliminaries in Section 2. Then, in Section 3, we present an overview of our WALK-ESTIMATE algorithm, and outline the key technical challenges for this design. In Sections 4 and 5, we develop the two main steps, WALK and ESTIMATE, respectively. We discuss in Section 6 two related issues: one is the application of our idea to another way of performing random walks, i.e., the ``one-long-run'' scheme. The other is the limitation of our techniques. We present the experimental results in Section 7, followed by a discussion of related work in Section 8 and the final remarks in Section 9.

%% file: 2-preliminaries.tex
\section{Preliminaries}
\label{sec:prelims}

\subsection{Graph Model}
\label{subsec:graphModel}

In this paper, we consider online social networks with the underlying topology of an undirected graphs $G \langle V, E\rangle$, where $V$ and $E$ are the sets of vertices and edges, respectively. Note that for online social networks which feature directed connections (e.g., Twitter), a common practice in the literature (e.g., \cite{kwak2010twitter}) is to reduce it to an undirected graph by defining two vertices $v_1, v_2 \in V$ to be connected in the undirected graph if and only if both $v_1 \to v_2$ and $v_2 \to v_1$ exist as directed connections. We use $|E|$ to denote the number of edges in the graph. For a given node $v \in V$, let $N(v)$ be the set of neighbors of $v$, and $d(v) = |N(v)|$ be its degree.

The web interface of the online social network exposes a restricted access interface which only allows {\em local neighborhood queries}. That is, the interface takes as input a node $v \in V$, and outputs $N(v)$, the set of $v$'s neighbors. The objective of sampling, as mentioned in the introduction, is to generate a sample of $V$ (according to a pre-determined sampling distribution) by issuing as few queries through the restrictive access interface as possible.

\subsection{Traditional Random Walks}
\label{subsec:traditionalRW}

\subsubsection{Overview}
A random walk is a Markov Chain Monte Carlo (MCMC) method on the above-described graph $G$. Intuitively, all random walks share a common scheme: it starts from a {\em starting node} $v_0 \in V$. At each step, given the current node it resides on, say $v_i$ for the $i$-th step, the random walk randomly chooses its next step from $v_i$'s immediate neighbors and $v_i$ itself (i.e., self-loop might be allowed) according to a pre-determined distribution, and then transits to the chosen node (one can see that $v_{i+1} \in \{v_i\} \cup N(v_i)$). We refer to this distribution over $N(v_i)$ as the {\em transit design}. As discussed below in examples of random walks, existing random walk designs often choose either a fixed distribution (e.g., uniform distribution), or a distribution determined by certain measurable attributes (e.g., degree) for nodes in $N(v_i)$.  One can see that the transition design can be captured by a $|V| \times |V|$ transition matrix $T$, in which $T_{ij}$ is the probability for the random walk to transit to node $v_j$ if its current state is $v_i$. 

Let $p_t(u)$ be the {\em sampling probability} for a node $u \in V$ to be taken at Step $t$ of the random walk (i.e., $p_t(u) = \Pr\{u = v_t\}$). A special property of random walk, which makes it suitable for our purpose of node sampling, is that as long as the graph $G$ is irreducible \cite{Gilks99}, $p_t(u)$ always converges to a fixed distribution when $t \to \infty$, no matter what the starting node $v_0$ is for the random walk. 

\subsubsection{Examples}

There are many different types of random walks  according to different designs of transition matrix $T$. We use {\em Simple Random Walk (SRW)} and {\em Metropolis-Hastings Random Walk (MHRW)} because of their popularity in the study of sampling online social networks.

\begin{definition} {\bf (Simple Random Walk (SRW))}. Given graph $G \langle V, E\rangle$, and a current node $u\in V$, a random walk is called Simple Random Walk if it uniformly at random chooses a neighboring node $v$ from $u$'s neighbors as the next step. The transition matrix $T$ is
\begin{equation}
\label{SRW-transition}
    T(u, v)=\left\{
    \begin{array}{ll}
        1/|N(u)|\,\,\,&\text{if $v \in N(u)$,}\\
        0\,\,\,&\text{otherwise.}
    \end{array}
    \right.
\end{equation} 
\end{definition}

\begin{definition} {\bf (Metropolis-Hastings Random Walk (MHRW))}. Given graph $G \langle V, E\rangle$, and a current node $u\in V$, a random walk is called Metropolis-Hastings Random Walk if it chooses a neighboring node $v$ according to the following transition matrix $T$:
\begin{equation}
\label{MH-transition}
T(u, v) =
 \begin{cases}
  \frac{1}{|N(u)|}.\min\{1,\frac{|N(u)|}{|N(v)|}\} & \text{if } v \in N(u) \\
   1 - \sum _{w \in N(u)}T(u,w)        & \text{if } u=v\\
    0 & \text{otherwise }
  \end{cases}
\end{equation}
We note that we set the target stationary distribution as uniform distribution for MHRW.
\end{definition}

\subsubsection{Burn-In Period}

With the traditional design of a random walk, its performance is determined by how fast the random walk converges to its stationary distribution, because only after so can the random walk algorithm takes a node as a sample. To capture this performance measure, {\em burn-in period} is defined as the number of steps it takes for a random walk to converge to its stationary distribution, as shown in the following definition.

\begin{definition} {\bf (Relative Point-wise Distance)}. Given graph $G \langle V, E\rangle$, \madd{and positive number of steps t}\del{and after $t$ steps}, Relative Point-wise Distance is defined as the following distance between the stationary distribution and the probability distribution for nodes to be taken at Step $t$:
    \begin{align}
    \bigtriangleup(t)=\mathop{\max}_{u, v\in V, v\in N(u)}\left\{\frac{|T_{uv}^{t}-\pi(v)|}{\pi(v)}\right\} 
    \end{align}
where \del{ $T^t(u,v)$} \madd{$T_{uv}^{t}$} is the element of $T^t$ (transition matrix $T$ to the power of $t$) with indices $u$ and $v$\madd{, and $\pi$ is the stationary distribution of the random walk}~\cite{Jerrum:1988:CRM:62212.62234}. The {\bf burn-in period} of a random walk is the minimum value of $t$ such that $\bigtriangleup(t) \leq \epsilon$ where $\epsilon$ is a pre-determined threshold on relative point-wise distance.
\end{definition} 

In practice, a popular technique for checking (on-the-fly) whether a random walk has reached its stationary distribution is called the {\em convergence monitors} (i.e. MCMC convergence diagnostics) \cite{Levin:2008}. For example, the Geweke method (summarized in \cite{Cowles96markovchain}) considers two ``windows'' of a random walk with length $l$: Window A is formed by the first 10\% steps, and Window B is formed by the last 50\%. According to \cite{Das:2013:FRW:2510649.2511334}, if the random walk indeed converges to the stationary distribution after burn-in, then the two windows should be statistically indistinguishable. Let
\begin{align}
Z = \left|\frac{\bar{\theta}_A -\bar{\theta}_B }{\sqrt{\hat{S}_\theta^A + \hat{S}_\theta^B }}\right|,
\end{align}
where $\theta$ is the attribute that can be retrieved from nodes (a typical one is the degree of a node), and $\bar{\theta}_A, \bar{\theta}_B$ are means of $\theta$ for all nodes in Windows $A$ and $B$, respectively, and $S_\theta^A$ and $S_\theta^B$ are their corresponding variances. One can see that $Z \to 0$ when the random walk converges to the stationary distribution. We use Geweke method as the convergence monitor in the experiments, and we set the threshold to be $Z \leq 0.1$ by default, while also performing tests with the threshold $Z \leq 0.01$.

A property of the graph which has been proven to be strongly correlated with the burn-in period length is the {\em spectral gap} of the transition matrix $T$. We denote the spectral gap as $\lambda = 1 - s_2$ where $s_2$ is the second largest eigenvalue of $T$.

\subsection{Acceptance-Rejection Sampling}
\label{subsec:acceptanceRejectionSampling}

Acceptance-rejection sampling (hereafter referred to as rejection sampling) is a technique we use to ``correct'' a sampling distribution to the desired target distribution. To understand how, consider the case where our algorithm samples a node with probability $p(u)$, while the desired distribution assigns probability $q(u)$ to node $u$. In order to make the correction, we take as input a node $u$ sampled by our algorithm, and ``accept'' it as a real sample with probability
\begin{align}
\beta(u) = \frac{q(u)}{p(u)} \cdot \min_{v \in V} \frac{p(v)}{q(v)},
\end{align}
because after such correction, the probability distribution of node $u$ in the final sample is
\begin{align}
\frac{p(u) \cdot \beta(u)}{\sum_{u \in V} (p(u) \cdot \beta(u))} = q(u),
\end{align}
which conforms to the desired target distribution.

A practical challenge one often faces when applying rejection sampling is the difficulty of computing $\min_v p(v)/q(v)$, especially when the graph topology is not known beforehand. \mdel{As we shall further elaborate in Section~\ref{subsec:idealWalk}, e}Even when a theoretical lower bound on $\min_v p(v)/q(v)$ can be computed, its value is often too small to support the practical usage of rejection sampling. A common practice to address this challenge is to replace $\min_v p(v)/q(v)$ with a manual threshold (e.g., \cite{Dasgupta:2007,Geyer1992}). Note that a large threshold might introduce bias to the sample - e.g., a threshold greater than $\min_v p(v)/q(v)$ would make the computed $\beta(u) > 1$ for certain nodes, essentially under-sampling them in the final sample. Nonetheless, such a large threshold also improves the efficiency of sampling, as the rejection probability will be lower. \madd{Of course, such an approximation can be made more conservatively (i.e., lower) to reduce bias, or more aggressively (i.e., higher) to make the sampling process more efficient.}
\del{As we shall further discuss in Section~\ref{subsec:idealWalk}, the threshold should be set to make a proper tradeoff between efficiency and sample bias according to the desired application.}

\subsection{Performance Measures}
\label{subsec:perfMeasures}

There are two important performance measures for a sampling algorithm over an online social network. One is its query cost - i.e., the number of nodes it has to access in order to obtain a predetermined number of samples. Note that query cost is the key efficiency measures here because many website enforce a query rate limit on the number of nodes one can access from an IP address or API account for a given time period (e.g., a day). As such, a random walk based sampling algorithm has to minimize the number of steps it takes to generate samples.

The other key performance measures here is the sample bias - i.e., the distance between the actual sampling distribution (i.e., the probability distribution according to which each node is drawn as a sample) and a predetermined target distribution. Note that while the uniform distribution is often used as the target distribution to ensure equal chance for all nodes, the target distribution can also have other values - e.g., proportional to the node degree (when simple random walk is used).

Another important issue with the definition of sample bias is the distance measure being used. Traditionally (e.g., in the studies of burn-in period and convergence monitoring for bounding the sampling bias), a popular measure is the vector norm for the difference between the two probability distribution vectors. For example, the {\em variation distance} measure the $\ell_\infty$-norm of the difference vector - i.e., the maximum absolute difference for the sampling probability of a node. While this is a reasonable measure for theoretical analysis, it can be difficult to use for experimental evaluations, as obtain the actual sampling probability for every node requires running the sampling algorithm repeatedly for extremely large number of times, especially when the underlying graph is large. To address the problem, in this paper, we use the vector norm measure for theoretical analysis, while using a different measure for experiments: specifically, we measure the error while using the obtained sample to estimate AVG aggregates such as the average degree of all nodes in the graph. We shall further elaborate the design of this experimental measure and the various AVG aggregates we use in the experimental evaluation section.

%% file: 4-walk.tex
\section{Overview of WALK-ESTIMATE}

In this section, we provide an overview of WALK-ESTIMATE, our main contribution of the paper. Specifically, we first describe the input and output of the algorithm, followed by a brief description of our key ideas and an outline of the main technical challenges for WALK and ESTIMATE, respectively.

\vspace{1mm}
\noindent{\bf Input \& Output:} The design objective of WALK-ESTIMATE is to achieve universal speed-up for MCMC sampling (random walks) over online social networks regardless of their target sampling distribution (and correspondingly, transition design). To achieve this goal, WALK-ESTIMATE takes as input (1) the transition design of an MCMC sampling algorithm, and (2) the desired sample size $h$. The output consists of $h$ samples taken according to the exact same target distribution as the input MCMC algorithm (subject to minimal sampling bias, as we shall further elaborate in latter sections). As discussed in Section 2, during this sampling process, WALK-ESTIMATE aims to minimize the query cost.

\vspace{1mm}
\noindent{\bf Key Ideas:} Recall from the introduction that our main novelty here is to forgo the long ``wait'' (i.e., burn-in period) required by traditional random walks, and instead WALK an optimal (much smaller) number of steps (often only a few steps longer than the graph diameter - see below for details). Of course, having a drastically shorter walk also makes our sampling distribution different from the target one we have to achieve at the end. To address this problem, our WALK calls upon the ESTIMATE component to estimate the probability for a node to be sampled by a (now much shorter) walk. Not that such estimated probability allows us to perform acceptance-rejection sampling \cite{Lovasz1993} over the nodes sampled in WALK, which eventually leads to samples taken according to the target distribution.

\vspace{1mm}
\noindent{\bf Technical Challenges for WALK:} One can see from the above description that the design of the two components face different challenges: For ``WALK'', the main challenge is how to properly determine the number of steps to walk. Clearly, the walk length must be at least the diameter of the graph in order to ensure a positive sampling probability for each node. On the other hand, an overly long walk not only diminishes the saving of queries, but might indeed cost even more queries than traditional random walks when the cost of ESTIMATE is taken into account. We shall address this challenge in Section~\ref{sec:walk} - and as we shall further elaborate there, fortunately, for real-world social networks, there is usually a wide range of walk lengths with which the WALK step can have a significant saving of query cost even after rejection sampling.

\vspace{1mm}
\noindent{\bf Technical Challenges for ESTIMATE:} For ESTIMATE, the key challenge is how to enable an accurate estimation for the sampling probability of a node without incurring a large query cost. Note that, after we repeatedly run WALK to generate (say 100) samples, there may be nodes sampled multiple times by WALK for which we can directly estimate their sampling probability (as their relative frequency within the collected sample). Nonetheless, for the vast majority of nodes which are sampled only once (almost always the case when the graph being sampled is large), it is unclear how one can estimate their sampling probabilities. We shall address this challenge in Section~\ref{sec:estimate} and show that (1) there is a surprisingly simple algorithm which enables a completely unbiased estimation of the sampling probability and consumes only a few extra queries, and (2) there are two effective heuristics which reduce the estimation variance even further, leading to more accurate estimations.

In the next two sections, we shall develop or techniques for the two components, WALK and ESTIMATE, respectively. The combination of them forms Algorithm WALK-ESTIMATE which, as we demonstrate in the extensive experimental results in Section~\ref{sec:exp}, produces higher quality (i.e., lower bias) samples than the traditional random walk algorithms while consuming fewer queries.

\section{WALK}
\label{sec:walk} 
We start with developing Algorithm WALK which significantly improves the efficiency of sampling by having a much shorter random walk followed by a rejection sampling process. Note that, for the ease of discussions, we separate out the discussion of sampling-probability estimation to the ESTIMATE component discussed in the next section - i.e., Algorithm WALK calls upon Algorithm ESTIMATE as a subroutine.

In this section, we first illustrate the key rationale behind our design with an ideal-case algorithm, IDEAL-WALK, and then present theoretical analysis which shows that a shorter walk followed by an acceptance-rejection procedure can almost always outperform the traditional random walk, no matter what the starting point is or the graph topology looks like. To study how much improvement a short walk can offer, we describe case studies with the underlying graph generated from various theoretical graph models. Finally, we conclude this section with the practical design of Algorithm WALK.

\subsection{IDEAL-WALK: Main Idea and Analysis}
\label{subsec:idealWalk}
The key rationale behind our idea of performing a short walk followed by acceptance-rejection sampling can be stated as follows. Recall from Section 2 that the long walk is required by traditional random walks to reduce the ``distance'' between its sampling distribution and the target, stationary, distribution - a distance often measured by the difference (e.g., $\ell_\infty$-norm) between the two probability vectors.

Consider how such a difference changes as the walk becomes longer. When the walk first starts, the sampling distribution is extremely skewed - i.e., $p_1(v) = 1$ on one node (the starting one) and 0 on all others - leading to an extremely large distance. As the walk proceeds, the distance decreases quickly - for example, as long as the walk length exceeds the graph diameter, all values in the sampling probability vector become positive\footnote{Note that here we assume each node has a nonzero (can be arbitrarily small) probability to transit to itself, to eliminate trivial cases where the graph is not irreducible.}, while the maximum value in the vector tends to decrease exponentially with the (initial few) steps taken - leading to a sharp decrease of the distance.

Nonetheless, it is important to note that the speed of reduction on the distance becomes {\em much slower} as the random walk grows longer. A simple evidence is the asymptotic nature of burn-in as discussed in Section 2 - which shows that, for some graphs, the ultimate reduction to zero distance never completes with a finite number of steps. Figure~\ref{fig:barabasi_max_min} demonstrates a concrete example for a random scale free network with $31$ nodes generated by the Barabasi-Albert model~\cite{Barabasi:1999}, where number of edges to attach from a new node to existing nodes is $3$. One can see from the figure that the speed of reduction declines sharply once the random walk grows longer than the graph diameter. In summary, one can observe the following ``behavior pattern'' of traditional random walks: to achieve a preset goal of shrinking the distance measure below a threshold, the random walk makes significant progress in the first few steps. Nonetheless, the ``benefit-cost ratio'' diminishes quickly as the random walk continues. As a result, a random walk might require a very long burn-in period to achieve the preset distance threshold.

\begin{figure}[ht]
	\centering
	\includegraphics[scale=0.3]{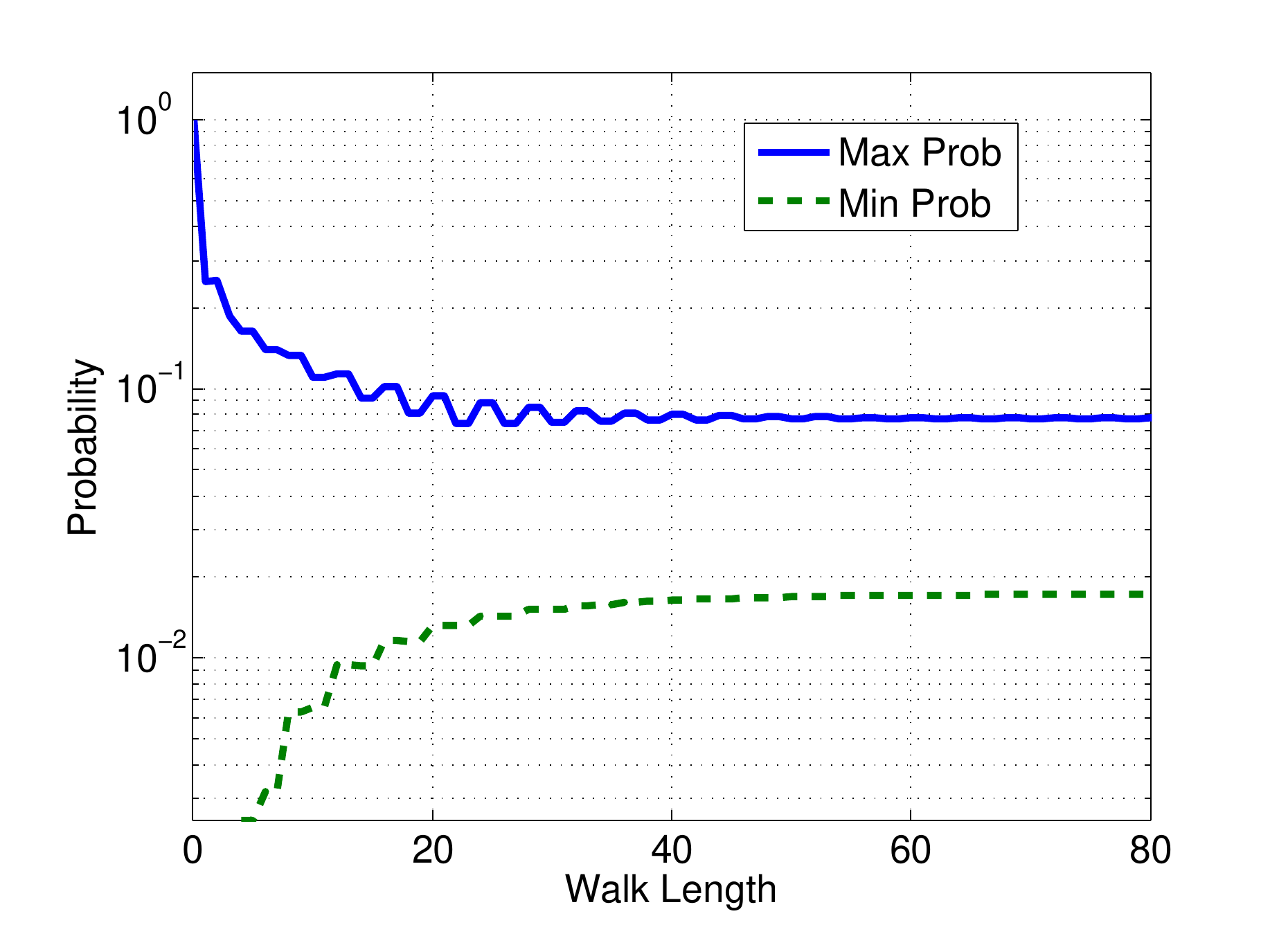}
	\vspace{-0.2in}
	\caption{Minimum and maximum probabilities vs walk length}
	\vspace{-0.1in}
	\label{fig:barabasi_max_min}
\end{figure}

Standing in sharp contrast to the above described behavior pattern is the performance of using acceptance-rejection sampling to achieve the pre-determined target distribution (instead of waiting for convergence). Interestingly, applying rejection sampling at the beginning of a random walk is often extremely costly - or even outright infeasible. For example, no rejection sampling can correct to a uniform target distribution before the walk is at least as long as the graph diameter. On the other hand, as the walk becomes longer, the cost of applying rejection sampling to reach the target distribution becomes much smaller - as we shall demonstrate as follows.

Consider an example where the target distribution is uniform. Note from acceptance-rejection sampling in Section~\ref{subsec:acceptanceRejectionSampling} that, in this case, the cost of rejection sampling is simply determined by (to be exact, inversely proportional to) the minimum value in the input sampling distribution - as the probability of accepting a sample is exactly the minimum probability multiplied by the number of nodes in the graph. As discussed above, this minimum probability grows from 0 at the very beginning to a positive value when the walk reaches the diameter, and often increases rapidly at the initial stage of random walk (see again Figure~\ref{fig:barabasi_max_min}). Correspondingly, the cost of rejection sampling drops significantly with a longer random walk.

One can observe from the above discussion an interesting distinction between two (competing) methods, (a) wait for the sampling distribution to converge to the stationary one, and (b) taking the current sampling distribution and directly ``correct'' it through rejection sampling: These two methods are better applied at different stages of a random walk process. Specifically, at the very beginning, method (b) is extremely costly or outright infeasible - so we should follow method (a) and walk longer for the sampling distribution to grow closer to the target vector. Nonetheless, after a certain number of steps, the direct correction (i.e., method (b)) becomes the better option because of the slower and slower convergence speed. Therefore, we should stop waiting for further convergence, and instead use rejection sampling to directly reach the desired distribution.

Of course, the above discussions leaves an important question unanswered: Given a reasonable threshold on the distance (between achieved sampling distribution and the desired stationary one), is there always a tipping point where we switch for waiting to correction? Note that the reason why the threshold value is important here can be observed from the extreme cases: when the threshold is extremely large, there is no need to switch because even the initial (one 1 and all other 0) distribution already satisfies the threshold. On the other hand, when the threshold trends to 0, as we discussed above, there are graphs for which the convergence length tends to infinity - i.e., it is always better to switch to rejection sampling as long as it has a finite cost. One can see from the extreme cases that whether switching to rejection sampling is effective in practice depends on whether the switch is necessary for {\em reasonable} thresholds that are just small enough to support real-world applications over the samples taken from online social networks. To this end, we have the following theorem.

\begin{theorem} \label{thm:mth}
Given an input random walk which has a transition design with spectral gap $\lambda$, to guarantee an $\ell_\infty$-variation distance of $\Delta$ between the sampling and target distributions, the expected query cost per sample of IDEAL-WALK which performs the random walk for $t_\mathrm{opt}$ steps followed by rejection sampling, where
\begin{align}
t_\mathrm{opt} = \frac{- \log (-\frac{1}{\Gamma} \cdot W(-\frac{\Gamma}{e d_{\max}}) \cdot d_{\max})}{\log (1-\lambda)}, \label{equ:gentva1}
\end{align}
where $d_{\max}$ is the maximum degree of all nodes in the graph and $W$ is the Lambert $W$-function, is always smaller than that of the the input random walk as long as $0 < \Delta < \Gamma$. Specifically, the ratio between the query cost per sample of IDEAL-WALK and the input random walk is at most
\begin{align}
& \frac{\Gamma \cdot t_\mathrm{opt} - t_\mathrm{opt} \cdot \Delta}{\Gamma - (1-\lambda)^{t_\mathrm{opt}} \cdot d_{\max}} \nonumber\\
\leq & \frac{- \log (-\frac{1}{\Gamma} \cdot W(-\frac{\Gamma}{e d_{\max}}) \cdot d_{\max})}{\log(\Delta/d_{\max})} \cdot \frac{\Gamma - \Delta}{\Gamma + \frac{\Gamma}{W(-\frac{\Gamma}{e d_{\max}})}}.
\end{align}
\end{theorem}

\begin{proof}
According to the $\ell_\infty$-norm mixing time of a Markov chain, if the random walk starts from $v$, we have a tight bound (tight in the worst-case scenario \cite{Zhu:2013})
\begin{align}
|p_t(u) - p(u)| \leq (1 - \lambda)^t \cdot \deg(v)
\end{align}

As such, to guarantee an $\ell_\infty$ variation distance of $\Delta$, the probability for rejection sampling to accept a sample taken after a walk of $t$ steps is at least
\begin{align}
\omega &\geq \frac{\Gamma - (1 - \lambda)^t \cdot \deg(v)}{\Gamma- \Delta}\\
&= \frac{\Gamma - (1 - \lambda)^t \cdot d_{\max}}{\Gamma- \Delta}.
\end{align}
Thus, in order to guarantee an $\ell_\infty$ variation distance of $\Delta$ in the worst-case scenario, the expected query cost per sample achieved by IDEAL-WALK is at most
\begin{align}
c \leq \min_{t: t > 0} \frac{t \cdot (\Gamma- \Delta)}{\Gamma - (1 - \lambda)^t \cdot d_{\max}}; \label{equ:genopt}
\end{align}
while the expected query cost per sample achieved by the input random walk is
\begin{align}
c_{\mathrm{RW}} = \frac{\log(\Delta/d_{\max})}{\log(1 - \lambda)}. \label{equ:genrw}
\end{align}

Note that for any $\Delta$, we always have $c \leq c_\mathrm{RW}$ because when $t = c_\mathrm{RW}$, we have
\begin{align}
\frac{t \cdot (\Gamma- \Delta)}{\Gamma - (1 - \lambda)^t \cdot d_{\max}} = c_\mathrm{RW}.
\end{align}
Intuitively, this simply means that by running IDEAL-WALK for as long as the input random walk (and therefore being able to skip rejection sampling), IDEAL-WALK is essentially reduced to the input random walk. Our task now is to determine whether $c < c_\mathrm{RW}$. We shall start with an intuitive discussion of why $c$ tends to be much smaller than $c_\mathrm{RW}$ for almost all realistic requirements of $\Delta$, and then present the formal analysis. Intuitively, consider the case where $\Delta$ is reduced by half, to $\Delta/2$. Note from (\ref{equ:genrw}) that the change of $c_\mathrm{RW}$ is always constant no matter how large (or small) $\Delta$ is - i.e.,$c_\mathrm{RW}$ will increase by $-\log 2 / \log(1 - \lambda)$. On the other hand, note from (\ref{equ:genopt}) that the change of $c$ is not always constant. Instead, the increase of $c$ is at most $t \cdot \Delta / (\Gamma - (1-\lambda)^t \cdot d_{\max})$ for $t$ which minimizes (\ref{equ:genopt}) for the original value of $\Delta$. Note that as $\Delta$ becomes smaller, $t$ either stays the same or becomes larger. As such, the smaller $\Delta$ is, the smaller the increase of $c$ will be. This is the intuitive explanation of why $c$ tends to be much smaller than $c_\mathrm{RW}$ when $\Delta$ is reasonably small.

Formally, we start with determining the optimal value of $t$ that minimizes (\ref{equ:genopt}). Let
\begin{align}
f(t) = \frac{t \cdot (\Gamma- \Delta)}{\Gamma - (1 - \lambda)^t \cdot d_{\max}}.
\end{align}

To satisfy $df(t)/dt = 0$, we have
\begin{align}
t &= \frac{(1-\lambda)^t \cdot d_{\max} - \Gamma}{\log(1-\lambda) \cdot (1-\lambda)^t \cdot d_{\max}}\\
&= \frac{1}{\log(1-\lambda)} \left(1- \frac{\Gamma}{(1-\lambda)^t \cdot d_{\max}}\right)
\end{align}

Thus, the optimal $t$ which minimizes $f(t)$ is
\begin{align}
t_\mathrm{opt} = \frac{- \log (-\frac{1}{\Gamma} \cdot W(-\frac{\Gamma}{e d_{\max}}) \cdot d_{\max})}{\log (1-\lambda)}, \label{equ:gentva2}
\end{align}
where $W$ is the Lambert $W$-function.

An interesting observation from (\ref{equ:gentva2}) is that $t_\mathrm{opt}$ is indeed irrelevant to $\Delta$. In other words, no matter how stringent (or loose) the requirement on $\Delta$ is, as long as $\Delta$ is smaller than $\Gamma$, IDEAL-WALK always outperforms the input random walk.

Finally, note that
\begin{align}
c \leq \frac{- \log (-\frac{1}{\Gamma} \cdot W(-\frac{\Gamma}{e d_{\max}}) \cdot d_{\max})}{\log (1-\lambda)} \cdot \frac{\Gamma - \Delta}{\Gamma + \frac{\Gamma}{W(-\frac{\Gamma}{e d_{\max}})}}.
\end{align}
Hence the query-cost ratio upper bound in the theorem.
\end{proof}

One can make an interesting observation from the proof of the theorem on how the performance of IDEAL-WALK changes when the walk length it takes grows larger. Initially, a larger $t$ leads to a smaller $c$, i.e., the expected query cost per sample for IDEAL-WALK, until $t$ reaches the optimal value $t_\mathrm{opt}$. Afterwards, a larger $t$ will lead to a larger $c$, until $c = c_\mathrm{RW}$ and IDEAL-WALK becomes equivalent with the input random walk. To understand the concrete values of $t_\mathrm{opt}$ and $c/c_\mathrm{RW}$, we consider a number of case studies in the following subsection.

\input{caseStudy}

\subsection{Algorithm WALK}

While the above theoretical analysis demonstrates the significant potential of query-cost savings by our WALK-ESTIMATE scheme, there is one key issue remaining before one can instantiate our idea into a practical WALK algorithm: in practice when the graph topology is not known beforehand, how can we determine the number of steps to walk before calling the ESTIMATE algorithm and performing the rejection sampling process? As one can see from the above discussions, an overly small length would lead to most samples being rejected, while an overly large one would incur unnecessary query cost for the WALK step.

Fortunately, we found through studies over real-world data (more details in the experimental evaluation section) that the setting of walk length is usually easy in practice as long as we set the walk length conservatively rather than aggressively. To understand why, note from the above case study, specifically the change of query cost per sample with walk length, that the query cost drops sharply before reaching the optimal walk length, the increase afterwards is much slower. As such, a reasonable strategy for setting the walk length is to be conservative rather than aggressive - i.e., giving preference to a longer, more conservative walk length. As we shall further elaborate in the experiments section, we use a default walk length of two times the graph diameter, which is conservatively estimated to be 10 for real world online social networks.

It is important to note that, while our experiments demonstrate that the above described heuristic strategy for setting the walk length works well over real world social networks, it is not a theoretically proven technique that works for all graphs. A simple counterexample here is the above-discussed Barbell graph - i.e., two complete graphs connected by one node, with one edge connected to each half. One can see that, while the graph has a very short diameter (i.e., 3), a random walk of length 6 is highly unlikely to cross to the other half of the graph, unless it starts from one of the three nodes that connect the two halves together. As such, the above heuristics for setting the walk length would yield an extremely small sample-acceptance probability and, therefore, a high query cost.

%% file: caseStudy.tex
\subsection{Case Study}
\label{subsec:caseStudy}
\begin{figure*}[t]
\begin{minipage}[t]{0.5\linewidth}
\centering
\includegraphics[scale=0.4]{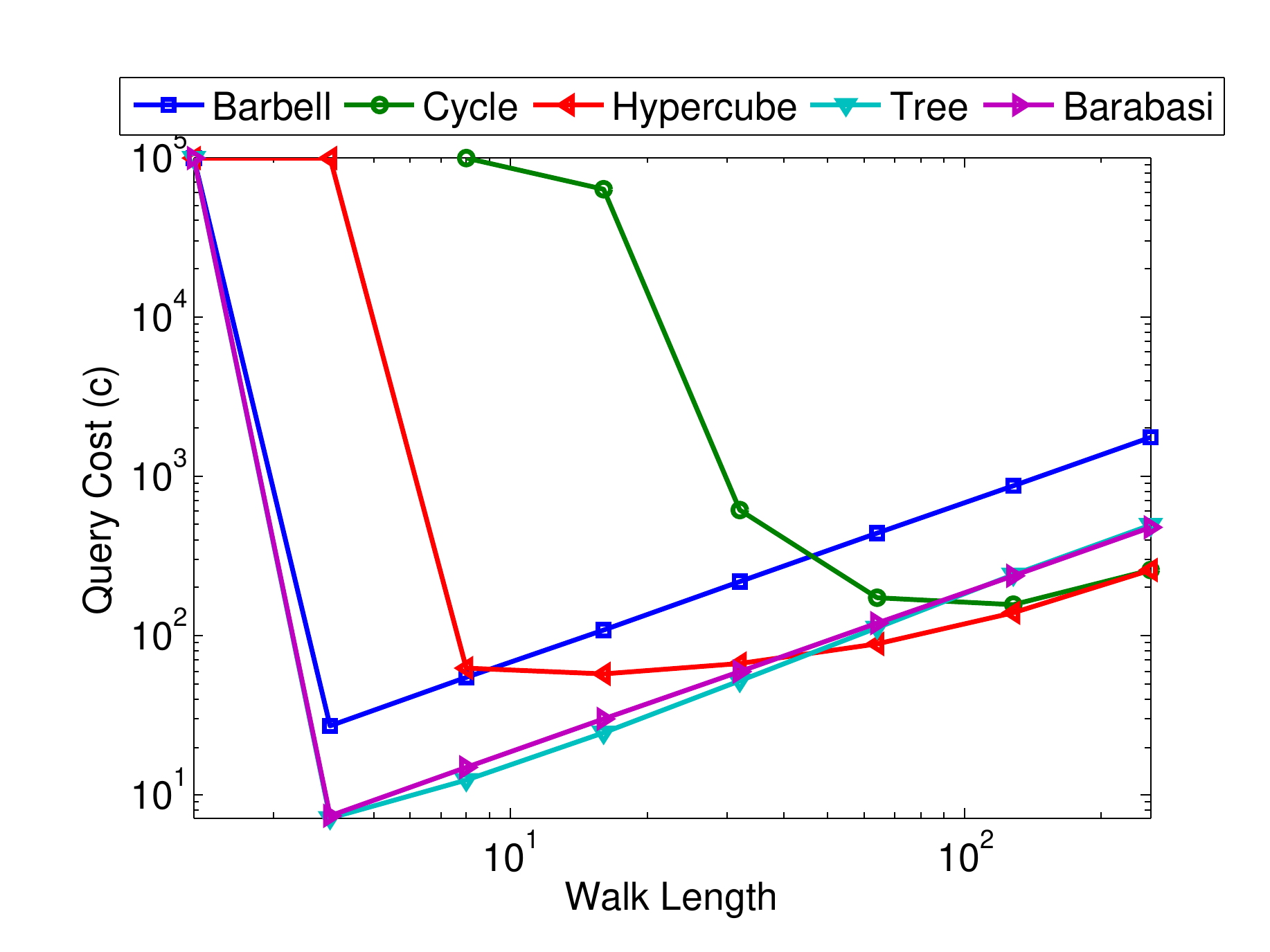}
\vspace{-0.15in}
\caption{Query cost per sample achieved by IDEAL-AR-SAMPLER.}
\label{fig:allCosts}
\end{minipage}
\hspace{0.5mm}
\begin{minipage}[t]{0.5\linewidth}
\centering
\includegraphics[scale=0.4]{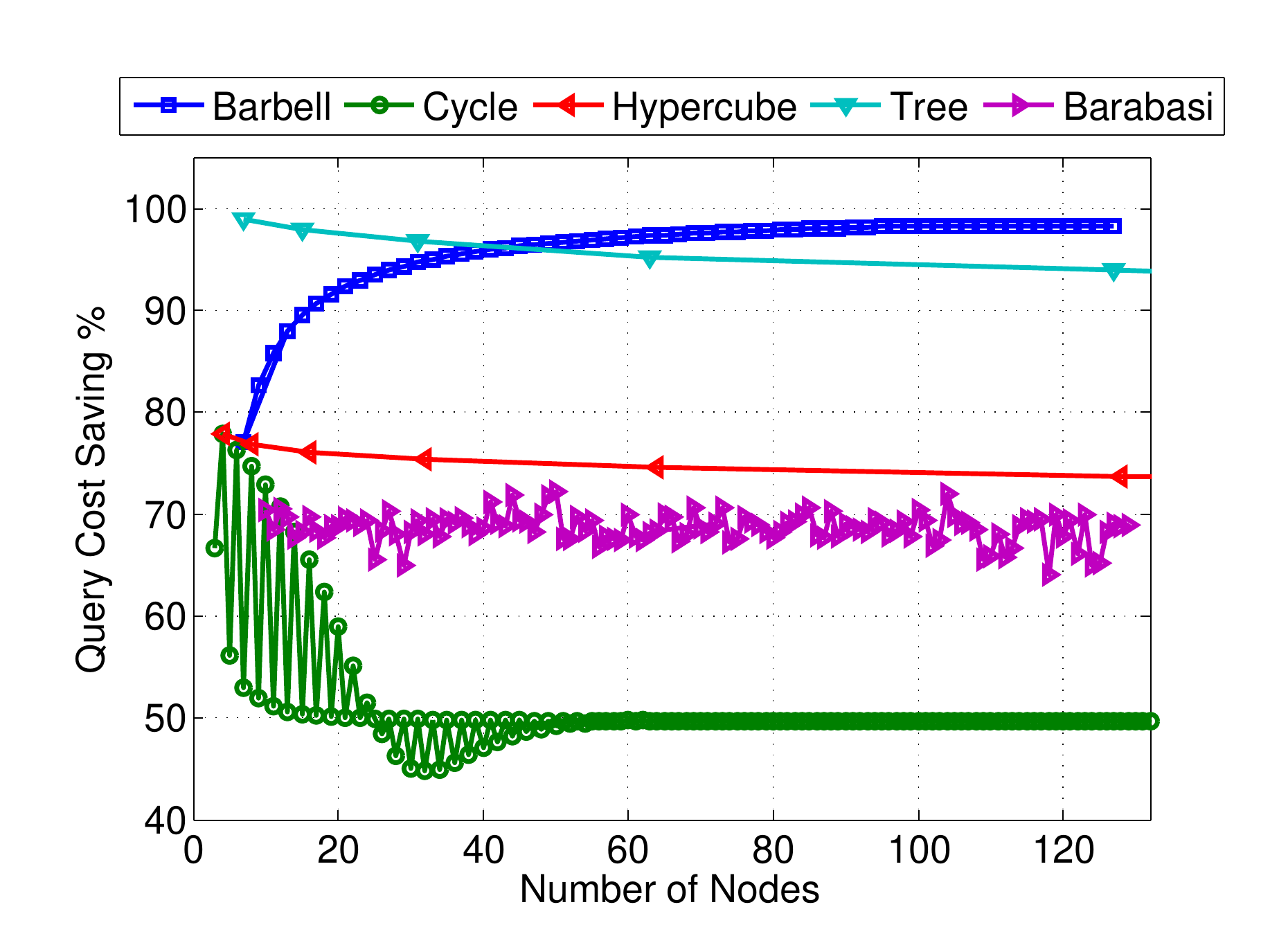}
\vspace{-0.15in}
\caption{ Query cost saving of the IDEAL-AR-SAMPLER.}
\label{fig:trend}
\end{minipage}
\end{figure*}

In this subsection, we compute numerically the values of $t_\mathrm{opt}$ and $c/c_\mathrm{RW}$ over a number of theoretical graph models, specifically Cycle, Hypercube, Barbell, Tree, and Barbasi-Albert (scale free) models: A cycle graph consists of a single circle of $n$ nodes - i.e., the graph has a diameter of $\lfloor \frac{n}{2} \rfloor$. A $k$-hypercube consists of $2^k$ nodes and $2^{k-1}k$ edges. If we represent each node as a (unique) $k$-bit binary sequence, and two nodes are connected if and only if their representations differ in exactly one bit. One can see that the hypercube has a diameter of $k$. A barbell graph of $n$ nodes is a graph obtained by connecting two copies of a complete graph of size $\frac{n-1}{2}$ by a central node, i.e, the diameter is $3$. A tree of height $h$ is a cycle free graph with at most $2^{h+1}-1$ nodes and diameter of $2h$. We considered a balanced binary tree, where the leaves are $h$ hops away
from the root. Finally, to simulate a scale-free network (with node degrees following a power-law distribution), we use the Barabasi-Albert model~\cite{Barabasi:1999}, where number of edges to attach from a new node to existing nodes is $3$. 

Figure~\ref{fig:allCosts} depicts how the expected query cost per sample changes when the length of walk taken by IDEAL-WALK varies from 1 to 128. We considered graphs with fix number of nodes $31$. Since hypercube should have $2^k$ nodes, we generate the one with $32$ nodes. In all cases, the target distribution is the uniform distribution. Note that if the walk length is smaller than the corresponding graph diameter, then we cost $c$ to be infinity. One can see from the figure that, for all graph models, the trend we observe from Theorem~\ref{thm:mth} holds - i.e., the query cost per sample $c$ drops dramatically at the beginning, reaches a minimal value, and then increases slowly. Another observation from the figure is that, in general, the larger diameter a graph has, the greater the optimal walk length for IDEAL-WALK will be. For example, compared with a Barbell graph with diameter of $3$, the cycle graph with diameter of $\lfloor \frac{31}{2} \rfloor = 15$ has a much longer optimal walk length - and consequently requires a larger query cost per sample.

Next, we examine the degree of improvement offered by IDEAL-WALK over the input random walk, again over the various graph models described above.  Figure~\ref{fig:trend} depicts how the ratio of improvement - i.e., $1 - c/c_\mathrm{RW}$ - changes when the graph size varies from 4 to 128. There are two interesting observations from the figure: One is that, while IDEAL-WALK offers over 50\% savings in almost all cases, the amount of savings does depend on the underlying graph topology - e.g., the improvement ratio is far smaller on cycle graphs than others, mainly caused by its large diameter and small spectral gap of $O(n^{-2})$~\cite{Levin:2008}.

The other observation is on how the improvement ratio changes with graph size: Interesting, when the graph becomes larger, the ratio increases for some models (e.g., Barbell), remains virtually constant for some others (e.g., hypercube, Barabasi-Albert), and declines for the ones left (e.g., cycle). An intuitive explanation here is that how the improvement ratio changes, as predicted in Theorem~\ref{thm:mth}, depends on a joint function of the graph size (e.g., $|E|$) and the spectral gap (i.e., $\lambda$). Since the spectral gap is difficult to directly observe, and there is a common understanding that the spectral gap is negatively correlated with the graph diameter \cite{Shyu:2013}, we illustrate the issue here by considering how the graph diameter changes with a linearly increasing node count for the various graph models: For the cycle graph, the diameter increases as fast as the node count - leading to a (generally) decreasing improvement ratio. For hypercube, tree and Barabasi-Albert models, the diameter increases at the log scale\footnote{To be exact, the diameter for Barabasi-Albert model is proportional to $\log n/\log\log n$~\cite{Cohen:2003}~\cite{bollobas2004diameter}} of node count - correspondingly, the improvement ratio is almost unaffected by the graph size. For Barbell graph, on the other hand, the diameter remains unchanged (i.e., 3) no matter how large the graph is. As a result, we observe a rapidly increasing improvement ratio from Figure~\ref{fig:trend}. Note that this is indeed a promising sign for the performance of IDEAL-WALK over real-world social networks, because it is widely believed that the diameter of such a network remains virtually constant (e.g.,~\cite{mislove2007measurement}~\cite{Kleinberg:2000}) no matter how large the graph size is - in other words, the improvement ratio offered by IDEAL-WALK is likely to increase as the graph becomes larger - a phenomenon we shall verify in the experiments section over the synthetic social networks by changing the size of the graph.

%% file: 5-estimate.tex
\section{ESTIMATE}
\label{sec:estimate}

Algorithm WALK leaves as an open problem of the estimation of sampling probability for a given node, so as rejection sampling can be applied to reach the input target distribution. In this section, we address this problem with Algorithm ESTIMATE. Specifically, we shall first describe a simple algorithm which, somewhat surprisingly, provides a completely unbiased estimation for the sampling probability with just a few queries. Unfortunately, we also point out a problem of this simple method: its high estimation variance which, despite the unbiasedness, still leads to a large error. To address this problem, we develop two heuristics, initial crawling and weighted sampling, to significantly reduce the estimation variance while requiring only a small number of additional queries.

\subsection{UNBIASED-ESTIMATE}
\label{sec:unbiasedEstimate}
\vspace{1mm}
\noindent{\bf Unbiased Estimation of Sampling Probability:} Recall that we use $p_t(u)$ to denote the probability for a node $u$ to be visited at Step $t$ of a random walk conducted by WALK, and $N(u)$ is the set of neighbors of $u$. To illustrate the key idea of UNBIASED-ESTIMATE, we start by considering the case where the input random walk is the simple random walk. One can see that
\begin{align}
p_t(u) = \sum_{u^\prime \in N(u)} \frac{p_{t-1}(u^\prime)}{|N(u^\prime)|}.
\end{align}

Thus, given $u$, a straightforward method of estimating $p_t(u)$ is to select uniformly at random a neighbor of $u$ (i.e., $u^\prime \in N(u)$), so as to reduce the problem of estimating $p_t(u)$ to estimating $p_{t-1}(u^\prime)$, because an unbiased estimation of $p_t(u)$ is simply
\begin{align}
\label{equ:bw}
\tilde{p}_t(u) = \frac{|N(u)|}{|N(u^\prime)|} \cdot p_{t-1}(u^\prime).
\end{align}

An important property of such an estimation is that as long as we can obtain an unbiased estimation of $p_{t-1}(u^\prime)$, say $\tilde{p}_{t-1}(u^\prime)$, then the estimation for $p_t(u)$ will also be unbiased. The reason for the unbiasedness can be stated as follows: Note that due to the conditional independence of the estimation of $\tilde{p}_{t-1}(u^\prime)$ with the selection of $u^\prime$ from $N(u)$, we have
\begin{align}
E(\tilde{p}_t(u)) &= \sum_{u^\prime \in N(u)} \frac{1}{|N(u)|} \cdot \frac{|N(u)|}{|N(u^\prime)|} \cdot E(\tilde{p}_{t-1}(u^\prime))\\
&= \sum_{u^\prime \in N(u)} \frac{1}{|N(u^\prime)|} \cdot E(\tilde{p}_{t-1}(u^\prime))\\
&= \sum_{u^\prime \in N(u)} \frac{1}{|N(u^\prime)|} \cdot p_{t-1}(u^\prime) = p_t(u),
\end{align}
where $E(\cdot)$ represents the expected value taken over the randomness of the estimation process. 

Given the unbiasedness property, we can run a recursive process for estimating $p_t(u)$ (with a decreasing subscript $t$) until reaching $p_0(w)$. Now we have $p_0(w) = 1$ if $w$ is the starting node of the random walk and 0 otherwise. One can see that this recursive process leads to an unbiased estimation of $p_t(u)$. We refer to this estimation method as UNBIASED-ESTIMATOR. The generic design of UNBIASED-ESTIMATOR (for any input MCMC random walk) is depicted in Algorithm~\ref{alg:unbiasedEstimate}, \madd{where $p_{uu'}$ is the element of the transition probability matrix in MCMC.}

\begin{algorithm}[!htb]
\caption{{\bf UNBIASED-ESTIMATE}}
\begin{algorithmic}[1]
\label{alg:unbiasedEstimate} 
\STATE {\bf Input: } Node $u$, Starting node $w$, Length of walk $t$ 
\STATE {\bf If} {$t = 0$ and $u$ == $w$} then {\bf return } 1 
\STATE {\bf If} {$t = 0$ and $u$ != $w$} then {\bf return } 0 
\STATE Choose a node $u'$ uniformly at random from $N(u)$
\RETURN \del{$\frac{|N(u)|}{|N(u')|}$} $|N(u)|$ $\cdot$ $p_{uu'}$ $\cdot$ UNBIASED-ESTIMATE($u'$, $w$, $t-1$)
\end{algorithmic}   
\end{algorithm}

\vspace{1mm}
\noindent{\bf Analysis of Estimation Variance:}  While the above UNBIASED-ESTIMATOR produces an unbiased estimations of the sampling probability, it also has an important problem: a high estimation variance which leads to a high estimation error (unless the estimator is repeatedly executed to reduce variance - which would lead to a large query cost nonetheless). Specifically, the estimation variance on the last few steps of the recursive process (i.e., with the smallest subscript in $p_t(u)$) are amplified significantly in the final estimation. To see this, consider a simple example of a $k$-regular graph. With UNBIASED-ESTIMATOR, the estimation of $p_t(u)$ is either $\tilde{p}_t(u) = 1$ (when the node $w$ encountered at $p_0(w)$ is the staring node) or 0 otherwise. As a result, the relative standard error for the estimation of $p_t(u)$ is exactly $\sqrt{(1-p_t(u))/p_t(u)}$. Since $p_t(u)$ is usually extremely small for a large graph, the relative standard error can be very high for UNBIASED-ESTIMATOR.

\input{huristics}

%\begin{align}
%(1 - 2|E| \cdot (1-\lambda)^t \cdot d_{\max}) = t (-2|E| d_{\max} (1-\lambda)^t \log(1-\lambda)).
%\end{align}

%Note that
%\begin{align}
%(1-\lambda)^t &= e^{1 - \frac{1}{2 |E| \cdot (1-\lambda)^t \cdot d_{\max}}}\\
%y &= e^{1 - \frac{1}{2 |E| \cdot y \cdot d_{\max}}}\\
%1/(1-z)/(2|E| d_{\max}) &= e^{z}\\
%1/h/(2|E| d_{\max}) &= e/e^{h}\\
%-(1/(2|E| d_{\max}))/e &= -he^{-h}\\
%h = -W(-1/(2|E| d_{\max}e))
%\end{align}

%\newpage
%$\mbox{ }$
%\newpage

%% file: huristics.tex
Our main idea for reducing the estimation variance is two-fold: {\em initial crawling} and {\em weighted sampling} - which we discuss in the next two subsections, respectively, before combing UNBIASED-ESTIMATE and the two heuristics to produce the practical ESTIMATE algorithm. 

\subsection{Variance Reduction: Initial Crawling}
\label{subsubsec:crawling}

Our first idea is to {\em crawl} the $h$-hop neighborhood of the starting point, so for each node $v$ in the neighborhood, we can precisely compute its sampling probability $p_t(v)$ for $t \leq h$. For example, if simple random walk is used in WALK, then all nodes $v$ in the immediate 1-hop neighborhood of starting node $s$ have $p_1(v) = \frac{1}{|N(s)|}$. In practice, $h$ should be set to a small number like $2$ or $3$ to minimize the query cost caused by the crawling process - note that the query cost is likely small because many nodes in the neighborhood may already be accessed by the WALK part, especially when multiple walks are performed to obtain multiple samples.

One can see that, with this initial crawling step, we effectively reduce the number of backward steps required by ESTIMATE because the backward estimation process can terminate as soon as it hits a crawled node. This shortened process, in turn, leads to a lower estimation variance and error.

\subsection{Variance Reduction: Weighted Sampling}
\label{subsubsec:weightedSampling}

Our second idea for variance reduction is weighted sampling - i.e., instead of picking $u^\prime$ uniformly at random from $N(u)$ as stated above (for estimating $p_t(u)$ from $p_{t-1}(u^\prime)$), we design the probability distribution based on the knowledge we already have about the underlying graph (e.g., through the random walks and backward estimations already performed). The key motivation behind this idea is the following observation on UNBIASED-ESTIMATE: When estimating $p_t(u)$, the values of $p_{t-1}(u^\prime)$ for all neighbors of $u$ (i.e., $u^\prime \in N(u)$) tend to vary widely - i.e., some neighbors might have much higher sampling probability than others. This phenomenon is evident from the fact that, even after reaching the stationary distribution of say the simple random walk, the sampling probability can vary by $d_{\max}/d_{\min}$ times, where $d_{\max}$ and $d_{\min}$ are the maximum and minimum degree of a node, respectively. On the other hand, without the initial crawling process, when $t = 1$, all but one neighbors of $u$ have $p_0(u^\prime) = 0$, while the other one has $p_0(u^\prime) = 1$ - also a significant variation.

Given this observation, one can see that we should allocate the queries we spend according to the value of $p_{t-1}(u^\prime)$ rather than simply at a uniform basis - specifically, we should spend more queries estimating a larger $p_{t-1}(u^\prime)$, simply because its estimation accuracy bears more weight on the final estimation error of $p_t(u)$. To this end, we adjust the random selection process of $u^\prime$ from $N(u)$ to the following {\em weighted sampling} process: First, we assign a minimum sampling probability $\epsilon$ to all nodes in $N(u)$ - to maintain the unbiasedness of the estimation algorithm. For the remaining $1 - \epsilon$ probability, we assign them proportionally to the total number of historic random walks which hit node $u^\prime$ at Step $t - 1$. \madd{More specifically, during the estimation process, all our random walks start from the same starting node. Suppose we have performed $n_{hw}$ random walks and currently performing the next one. Also suppose that we are at node $u$ at step $t$.  Let $u'$ be a neighbor of $u$ (i.e $u' \in N(u)$). Among the $n_{hw}$ random walks, we compute the number of times $u'$ is reached at step $t-1$. Let it be $n_{u',t-1}$. i.e. $0 \leq n_{u',t-1} \leq n_{hw}$. 
The ratio $\frac{n_{u',t-1}}{n_{hw}}$ has some impact on how often node $u'$ is picked as part of the random walk.}
\del{We refer the total number of historic walks as $n_{hw}$ and the number of historic random walks that hit node $u^\prime$ at Step $t-1$ as $n_{u', t-1}$.} Algorithm~\ref{alg:ws-bw} depicts the pseudocode of this weighted sampling scheme.

\begin{algorithm}[!htb]
\caption{{\bf WeightedSamplingBackward (WS-BW)}}
\begin{algorithmic}[1]
\label{alg:ws-bw} 
\STATE {\bf Input: } Node $u$, starting node $w$, Length of walk $t$, $\epsilon$
\STATE {\bf if } $t = 0$ and $u = w$ then {\bf return } $1$
\STATE {\bf if } $t = 0$ and $u \neq w$ then {\bf return } $0$
\STATE $\forall u' \in N(u)$, $\pi_{u'} = \epsilon / |N(u)|$
\STATE $\forall u' \in N(u)$, $\pi_{u'} = \pi_{u'} + (1 - \epsilon) (n_{u',t-1}/{n_{hw}})$ 
\STATE Choose node $v$ from $N(u)$ according to distribution $\pi$
\RETURN { $\frac{|N(u)|}{|N(v)|} \times$ WS-BW($v,w,t-1$)}
\end{algorithmic}   
\end{algorithm}

\subsection{Algorithm ESTIMATE}

We now combine UNBIASED-ESTIMATE with our two heuristics for variance reduction, initial crawling and weighted sampling, to produce Algorithm ESTIMATE. Note that there is one additional design in ESTIMATE which aims to further reduce the estimation error: For each $p_t(u)$ we need to estimate, we can repeatedly execute ESTIMATE (and take the average of estimations) to reduce the estimation error. The number of executions we take, of course, depends on the overall query budget. In addition, instead of running the same number of executions for all $u$, we next allocate the budget, once again, according to the estimations we have obtained so far for all nodes to be estimated. Specifically, we assign the number of executions in proportion to the estimation variance for each node. Figure~\ref{alg:estimate} depicts the pseudo code of Algorithm ESTIMATE.
%Terminated the walk if it hits one of the crawled node at step $i\leq (t-h)$.

%\begin{algorithm}[!htb]
%\caption{{\bf ESTIMATE}}
%\begin{algorithmic}[1]
%\label{alg:estimate} 
%\STATE {\bf Input: } Set of nodes $V$, Starting node $w$, length of walk $t$, number of crawling steps $h$
%\STATE Crawl $h$-hop neighborhood of $w$ and compute their exact landing probability of visited nodes by Equation~\ref{equ:exactProb}.
%\STATE Run \textcolor{red}{$n$} different random walks of length $t$ starting from $w$ to estimate the hitting probability of the nodes at different steps ($FW$).
%\FOR {$u \in V$} 
%\STATE For fix number of times run the Algorithm~\ref{alg:ws-bw}, WS-BW($u,w,t,FW$)(terminated the walk if it hits one of the crawled node at step $i\leq (t-h)$). 
%\STATE Find the variance of those estimations.
%\ENDFOR
%\WHILE {Remaining query budget $> 0$}
%\STATE For those with high varience run the Algorithm~\ref{alg:ws-bw}, WS-BW($u,w,t,FW$).
%\ENDWHILE   
%\end{algorithmic}   
%\end{algorithm}

\begin{algorithm}[!htb]
\caption{{\bf ESTIMATE}}
\begin{algorithmic}[1]
\label{alg:estimate} 
\STATE {\bf Input: } Starting node $w$, length of walk $t$, number of crawling steps $h$, forward random walks issued $F$%, Budget $B$, minimum budget $b$ for each node
\STATE Crawl $h$-hop neighborhood of $w$ and compute their exact sampling probability 
\STATE Let $V_F$ be the set of nodes hit by random walks in $F$
\FOR {$u \in V_F$} 
%\WHILE{$b \neq 0$}
\STATE $p_t(u)$ = WS-BW($u,w,t$)
%\STATE estimations($u$).add($p_t(u)$)
%\ENDWHILE
\STATE Compute estimation variance of estimations of $p_t(u)$%)%for $u$ as $\sqrt{(1-p_t(u)) / p_t(u) }$ 
\ENDFOR
\STATE Use remaining budget to reduce variance by invoking Algorithm~\ref{alg:ws-bw}. Choose nodes randomly proportional to their variance.
\end{algorithmic}   
\end{algorithm}

%% file: 6-discussion.tex
\section{Discussions}

\subsection{Many Short Runs v.s. One Long Run}

\madd{Broadly speaking, there are two major ways in which random walk has been used in existing literature for sampling purposes, i.e., ``many short runs'', and ``one long run''. Figure~\ref{fig:onemany} illustrates the difference between the two ways.
In ``many short runs'', the random walk repeatedly starts from a specific node until burn-in occurs 
and take a single sample from each walk. This is by far the most common way of using random walk as it produces i.i.d samples which produce superior estimates. In addition, many short runs can be easily embedded into parallel computing applications, and we can use multiple starting points in practice. According to \cite{Gilks99}, by taking a number of parallel replications and actively monitoring samples generated from multiple runs, we can guard against a single chain leaving a ``significant proportion'' of the sample space unexplored. 
In this paper, we compare our algorithm against this common variant. However, note that there is no chance of amortization here as we perform a new walk for each sample. 

The other way is to perform ``one long run'' where it first goes through the burn-in period for convergence to stationary distribution.
Once the burn-in period is over, the long run continues the walk and collects {\em every single node} 
encountered after the burn-in period into the sample pool.
This approach does indeed amortize the cost of burn-in as multiple samples are obtained after burn-in period.
However, this approach is not as commonly used as it produces dependent (correlated) samples.
When one uses the sample pool generated by one long run for purposes such as aggregate estimations (e.g., for AVG degree), it may be significantly less effective (resp.~less accurate) than a much smaller sample set produced by many short runs, especially when there is a strong correlation between the attribute values being aggregated on adjacent nodes (e.g., when nodes with larger degrees tend to be connected with each other). Indeed, a key concept capturing such a difference is the {\em effective sample size} \cite{301511} of one long run:
\begin{equation}
M = \frac{h}{1+2\sum_{k=1}^\infty\rho_k},
\end{equation}
where $h$ is the original sample size, and $\rho_k$ is the autocorrelation at lag $k$ (i.e., between the values of attribute being aggregated on nodes taken $k$ hops apart).

One can see from the above discussions that one long run is not a silver bullet solving the challenge of burn-in cost. Instead, it {\em might} be applied on cases when we know the intended application - specifically, the attribute to be aggregated - {\em and} such an attribute features a small autocorrelation. Indeed, our central contribution in the paper is a novel mechanism
to speed up the ``many short runs'' variant so that it could obtain independent samples at a much lower query cost. While we do observe the potential of applying our WALK-ESTIMATE idea to one long run - e.g., by estimating the sampling probability for not only the last node (taken as a candidate) but every node on the walk path, we leave the detailed investigation to further work.

\begin{figure}[t]
\centering
\includegraphics[scale=0.4]{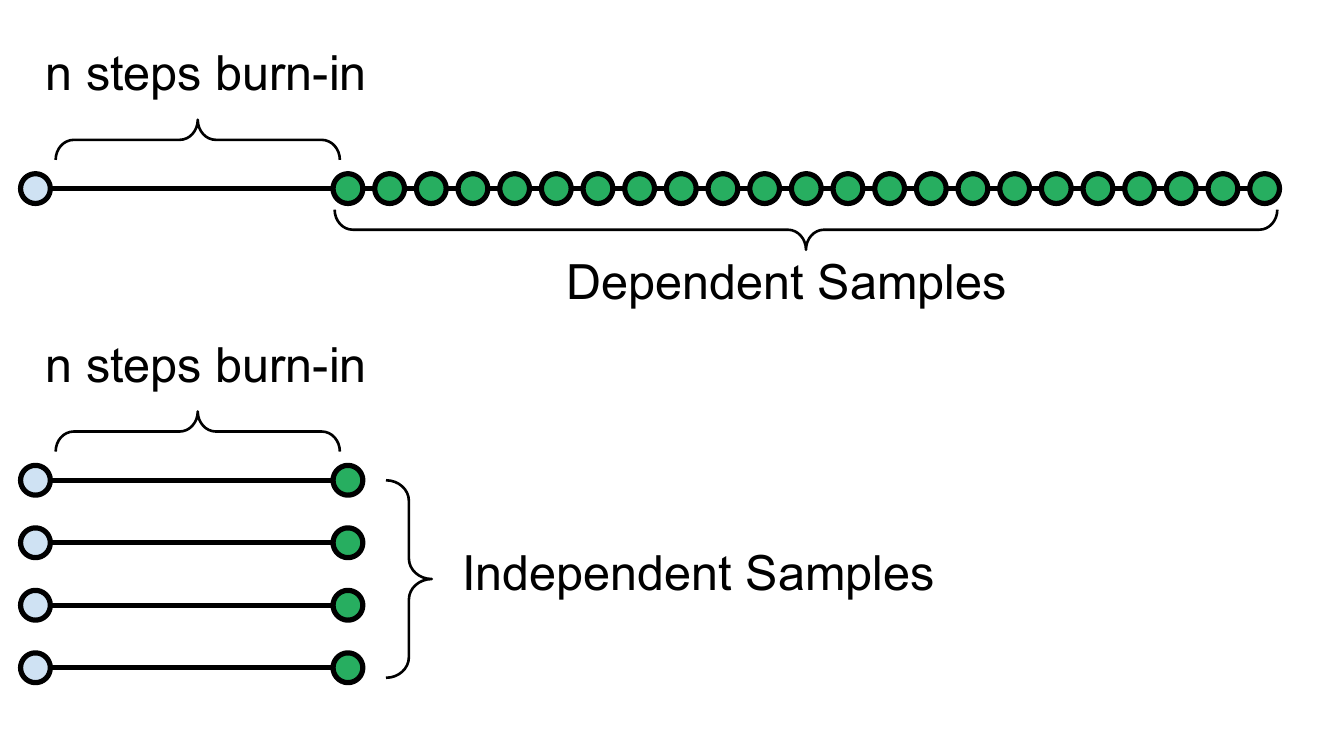}
\vspace{-0.2in}
\caption{Illustration of ``one long run'' v.s. ``many short runs''}
\vspace{-0.2in}
\label{fig:onemany}
\end{figure}
}

\subsection{Limitations of WALK-ESTIMATE - Graph Diameter}

\begin{figure}
\centering
\includegraphics[scale=0.28]{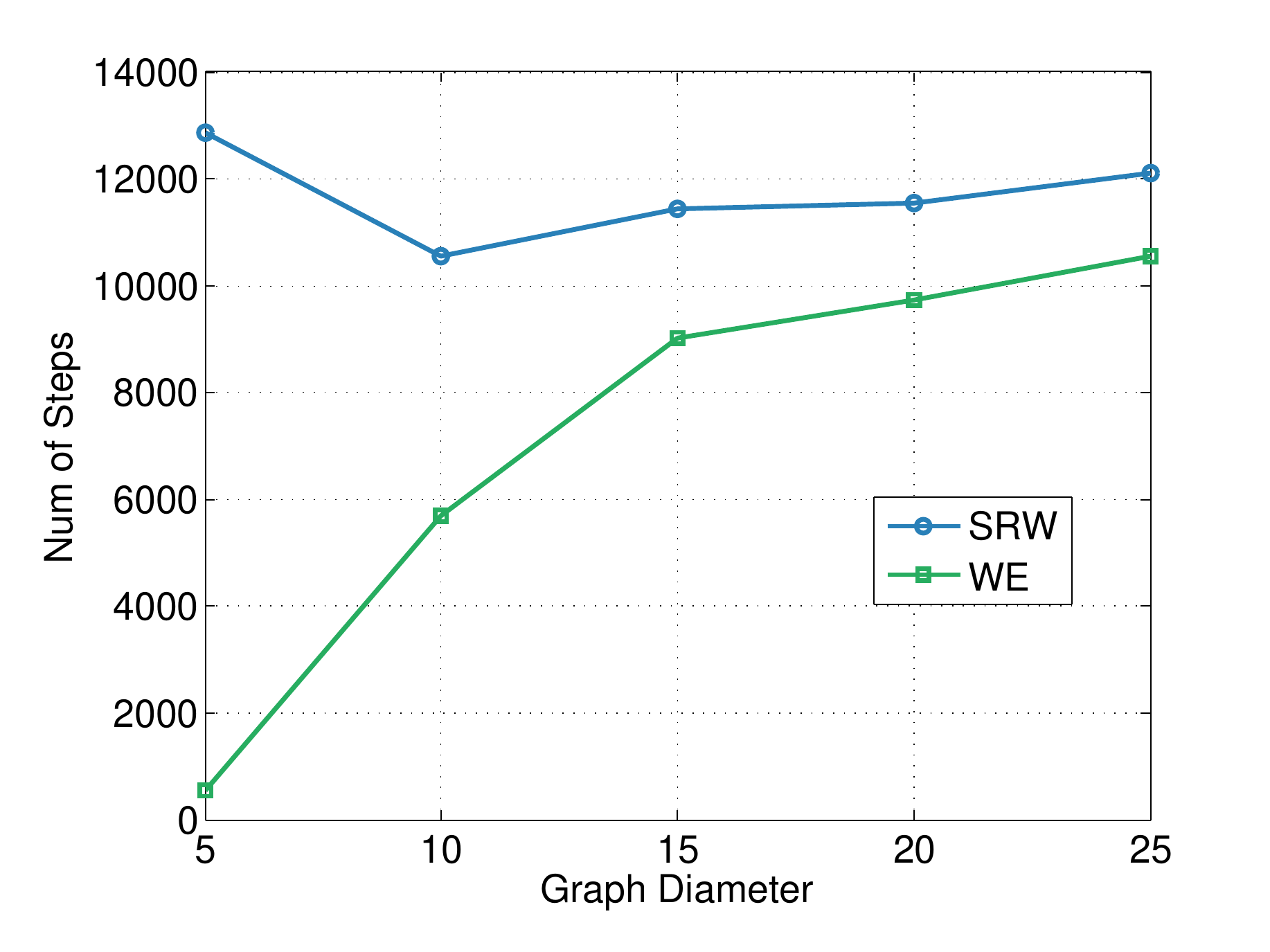}
\vspace{-0.2in}
\caption{Cycle graphs with long diameter}
\vspace{-0.2in}
\label{fig:cycle-diameter}
\end{figure}

Before concluding the technical discussion of WALK-ESTIMATE, we would like to point out its limitations - i.e., when it should {\em not} be applied for sampling a graph with local neighborhood access limitations. Specifically, we note that WALK-ESTIMATE should not be applied over graphs with long diameters. Note that while our results for IDEAL-WALK (Theorem~\ref{thm:mth}) appear to demonstrate efficiency enhancements regardless of the graph diameter, the performance of our ESTIMATE step is significantly worse when the graph diameter is large. The reason behind this is straightforward - in our backward walk for ESTIMATE, the probability of hitting the starting node (or the starting neighborhood crawled by the initial crawling process) decreases rapidly when the graph diameter becomes larger. This in turn leads to more backward walks being required for the estimations of sampling probability and, as a result, worst sampling efficiency. Figure~\ref{fig:cycle-diameter} demonstrates an example of how the average number of walk steps taken (for WALK-ESTIMATE, both forward and backward) per sample changes on cycle graphs for simple random walk (SRW) and our WALK-ESTIMATE algorithm (WE, with SRW as input) when the graph diameter grows from 5 to 25. The cycle graphs' sizes are $11, 21, \dots, 51$. One can see that unlike SRW which is barely affected by the growing diameter, the expected cost of WALK-ESTIMATE increases dramatically as the diameter becomes longer. Fortunately, it is important to understand that graphs with long diameters are {\em not} the intended target of this paper, because it is well known that online social networks, even the very large ones, have small diameters ranging from 3 to 8 \cite{Liben-nowell05analgorithmic}.

\subsection{Practical WALK-ESTIMATOR Challenges}
Indeed there are a number of practical issues that must be taken into account before a practical WALK-ESTIMATOR could be built.
The two core challenges are 1) access restrictions of the real world social networks such as rate limits, restricted access to neighbors etc, and 2) estimate of the scaling factor, i.e.,  $\min_{v \in V} (p(v)/q(v))$ in the acceptance-rejection sampling step.

\madd{
\subsubsection{Restrictions in Real-world Social Networks}
Real world social networks could place a number of access restrictions (such as rate limits, restricted access to neighbors etc).

{\bf Impact of Rate Limits:} 
However, most of those restrictions such as rate limits do not have a major impact on the estimation accuracy.
They could instead be treated as engineering challenges in building a practical system which is a well studied orthogonal issue.
There has been extensive literature 
\cite{catanese2011crawling,wondracek2010practical,chau2007parallel,bonneau2009prying,mislove2007measurement,nazir2008unveiling,ye2010crawling} 
that provide number of guidelines for crawling/sampling online social networks.
Since our experiments were conducted on well known benchmark datasets, these issues did not have any practical impact.
However, we plan to discuss the practical issues as part of future work on building a demonstration of our paper.

{\bf Impact of Access Restrictions to Neighbors:}
Access restrictions that limit how the neighbors of a node are obtained have some limited impact over algorithm.
However, it must be pointed out that, under some mild but realistic assumptions that hold for most real world social networks,
they do not have any significant impact over the accuracy of our algorithms.
Broadly speaking, access restrictions over neighbors take one of the following forms:
\begin{enumerate}
    \item The social network returns $k$ neighbors randomly during each invocation 
            (i.e. different invocation might see different $k$ neighbors)
    \item The social network returns $k$ fixed neighbors picked randomly (i.e. different invocation returns same $k$ neighbors)
    \item The social network bounds the maximum number of neighbors (say $l$) returned. 
\end{enumerate}

Twitter is one of the few social networks that provides a maximum of 5000 neighbors for an user - access restriction of type (3).
However, to the best of our knowledge, we have not seen any major social network with access restrictions (1) and (2).
Before describing the impact on our algorithm, 
we would like to note that, statistically, there is no distinction between access restrictions (2) and (3).
By setting the value of parameters $k$ and $l$, we can see that 
they provide identical interface to accessing the social network to a third party.

{\bf Impact of Restrictions of Type (1):}
Consider the scenario where, given a node $u$ in the graph, the API call $N(u)$ provides a random list of $k$ neighbors.
At each step, Simple Random Walk (SRW) seeks to choose one of $u$'s neighbor randomly. 
This is typically achieved by obtaining all neighbors $N(u)$ and choosing a node uniformly at random.
If the list of neighbors provided were already chosen uniformly at random, 
we could instead choose, say the first neighbor to traverse next.

There is a subtle issue in this setup. 
For example, if our objective is to estimate the average degree, we cannot directly use the number of neighbors returned as an estimate.
However, this issue can easily be circumvented using known techniques (such as mark-and-recapture \cite{robson1964sample,Katzir2011})
to estimate the degree of a graph by repeated invocation of neighbors API.
No other changes to our algorithm is required.

{\bf Impact of Restrictions of Type (2) and (3):}
If the neighbors API returns a fixed set of result, we cannot distinguish 
whether the API returns random or arbitrary subset of neighbors.
If the neighbors API returns a truncated list of neighbors, the key implication is that it limits the ``visible'' graph -
i.e. the partial subgraph that our algorithm could construct locally.
Consider the scenario where the neighbors API had returned all neighbors of $u$.
If $v$ was a neighbor of $u$, then we also know that $u$ is a neighbor of $v$ (and hence will be returned by $N(v)$).
However, under the truncated list of neighbors, there is no guarantee that 
if $v$ is returned via $N(u)$, then $u$ is returned in $N(v)$.

However, this issue could easily be handled by defining an alternate semantics for graph connectivity.
Specifically, when deciding if we can traverse an edge as part of random walk, we first perform a {\em bidirectional check}.
In other words, before traversing an edge $(u,v)$, we ensure that $u \in N(v) \wedge v \in N(u)$.
We do not use the edge otherwise.
While this might seem like a significant restriction, the actual impact on graph connectivity is limited
as long as the maximum size of neighbors returned by $N(\cdot)$ is not too small.
Even a value as small as 100 ensures graph connectivity and have negligible impact on the algorithms.
If the value is small, our algorithms will still provide unbiased estimates -
however, the graph may not be connected.

In summary, access restrictions such as rate limits are primarily an engineering issue that has been extensively studied.
Restrictions to neighbors do have some impact on our algorithm.
However, if the list of neighbors returned is not small, then the impact is negligible.
Finally, these restrictions affect both SRW and MHRW - and hence all the efficiency improvements that we propose 
are still effective even under this scenario.

\subsubsection{Estimating the scaling factor}
As we discussed in~\ref{subsec:acceptanceRejectionSampling}, the optimal scaling factor to maintain an unbiased sample is exactly $\min_{v \in V} (p(v)/q(v))$. However, in practice because of the lack of global topological knowledge of the real-world social networks, one may not be able to precisely compute $\min_{v \in V} (p(v)/q(v))$. A common technique used by acceptance-rejection sampling in statistics is to bootstrap an approximation of the $\min_{v \in V} (p(v)/q(v))$ based on the samples already observed, and then use such an approximation in acceptance-rejection sampling. Of course, such an approximation can be made more conservatively (i.e., lower) to reduce bias, or more aggressively (i.e., higher) to make the sampling process more efficient. In our experiments we derive the sampling probabilities of the visited nodes using the Algorithm~\ref{alg:unbiasedEstimate} and we consider the 10th percentile of the estimation of sampling probabilities as the $\min_{v \in V} (p(v)/q(v))$.
}

%% file: 7-exp.tex
\section{Experimental Evaluation}
\label{sec:exp}

\subsection{Experimental Setup}
\label{subsec:expSetup}

\begin{figure*}[ht]
	\centering
	\subcaptionbox{Average Degree (SRW)\label{fig:gplus_srw_relE_vs_QueryCost}}{
		\vspace{-0.05in}
		\includegraphics[scale=0.26]{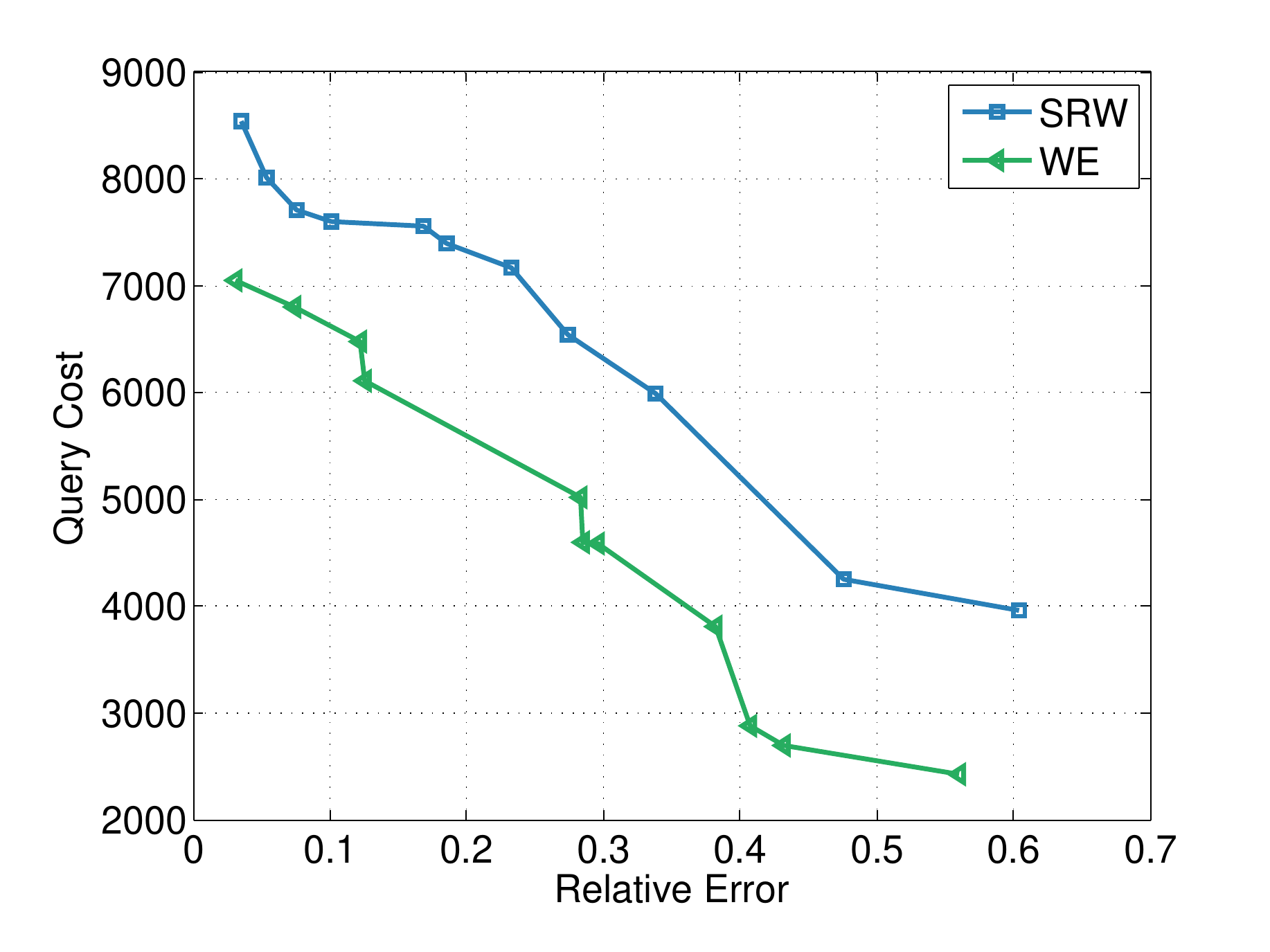}
		\vspace{-0.2in}
		\hspace{-0.4in}
	}
	\quad 
	\subcaptionbox{Average Self-description Length (SRW)\label{fig:gplus_srw_relE_vs_QueryCost_DL}}{
		\vspace{-0.05in}
		\includegraphics[scale=0.26]{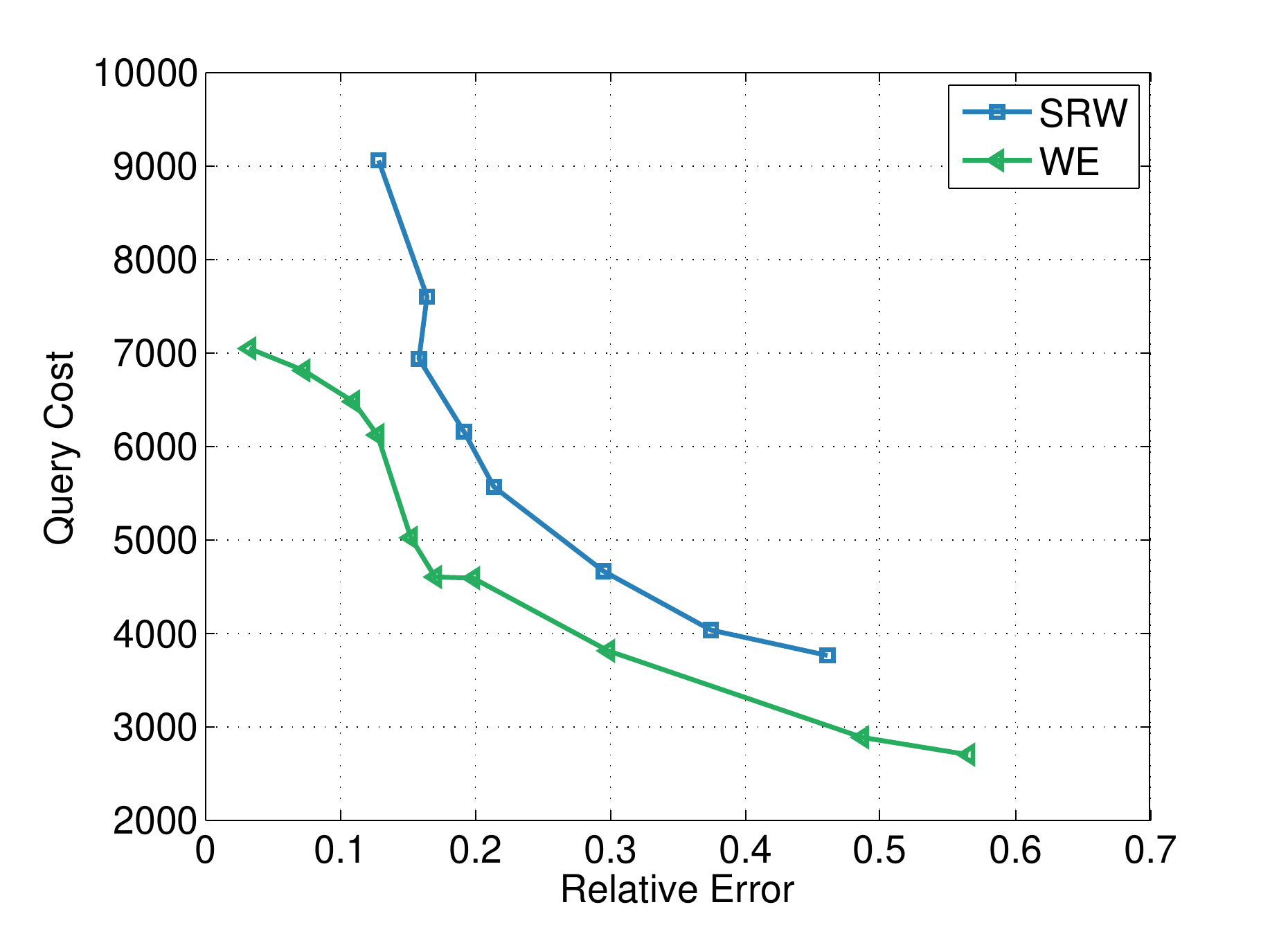}
		\vspace{-0.2in}
		\hspace{-0.4in}
	}
	\quad
	\subcaptionbox{Average Degree (MHRW)\label{fig:gplus_mhrw_relE_vs_QueryCost}}{
	\vspace{-0.05in}
		\includegraphics[scale=0.26]{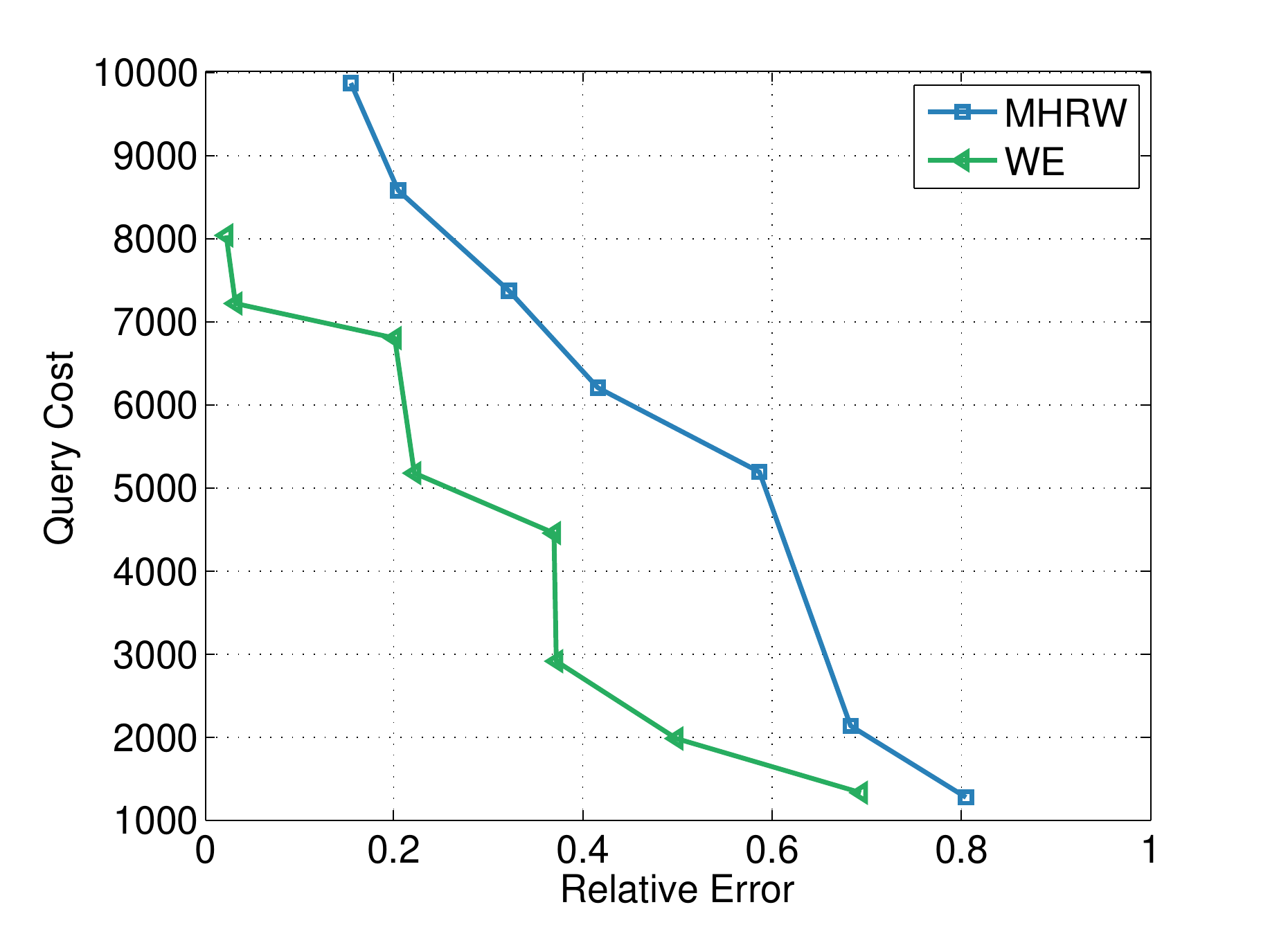}
		\vspace{-0.2in}
		\hspace{-0.4in}
	}
	\quad
	\subcaptionbox{Average Self-description Length (MHRW)\label{fig:gplus_mhrw_relE_vs_QueryCost_DL}}{
	\vspace{-0.05in}
		\includegraphics[scale=0.26]{pic/matlab_pic/Gplus-MHRW-RelativeError-QueryCost-Single}
		\vspace{-0.2in}
		\hspace{-0.4in}
	}
	\vspace{-0.1in}
	\caption{Relative Error of the Average Estimations vs Query Cost in Google Plus.}\label{fig:RelErrVSQC}
\end{figure*}
	%\vspace{-0.1in}

\begin{figure*}[ht]
        \quad
        \subcaptionbox{Yelp: Average Degree \label{fig:rev11YelpAvgDeg}}{
                \vspace{-0.05in}
                \includegraphics[scale=0.26]{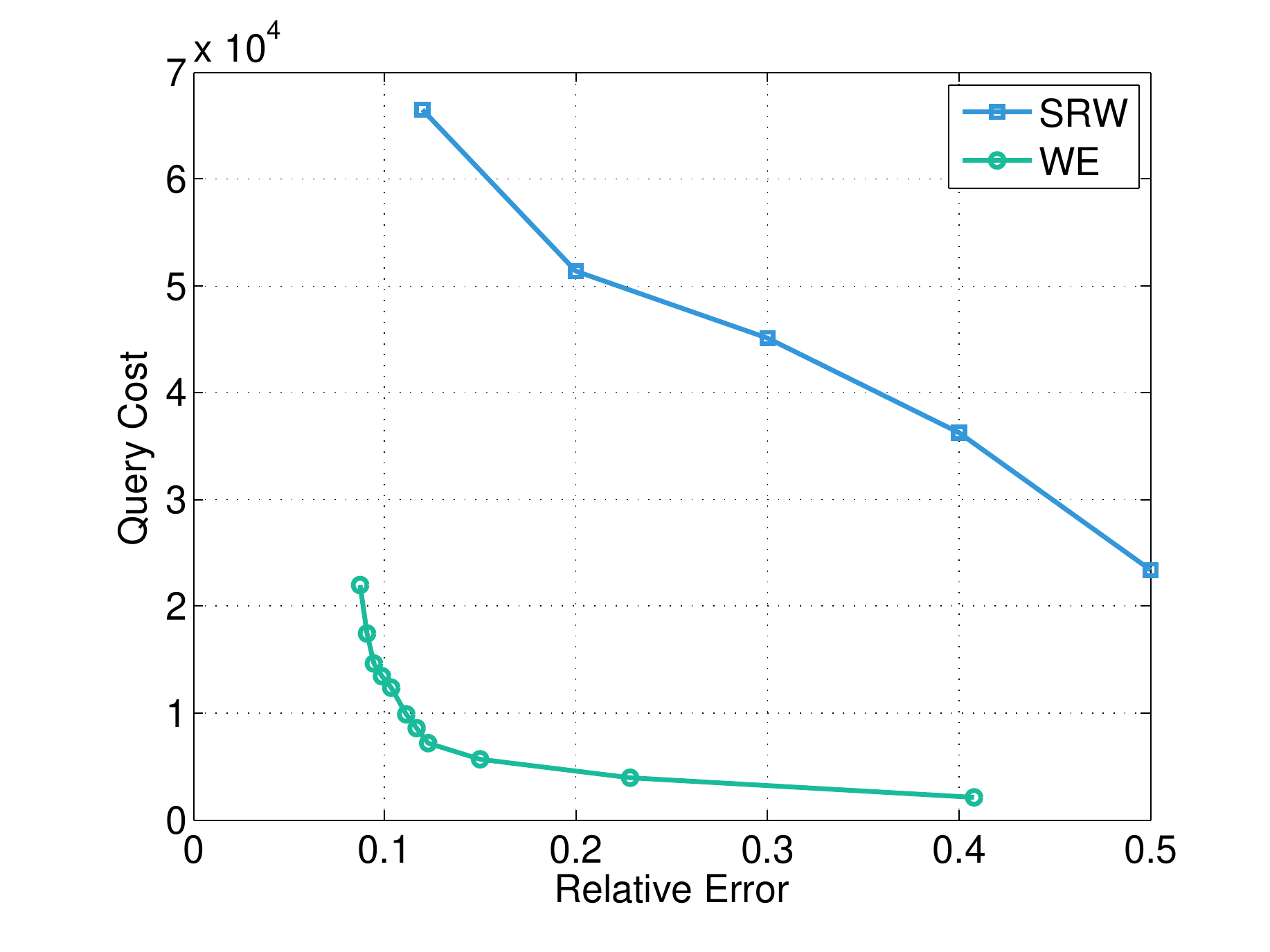}
                \vspace{-0.2in}
                \hspace{-0.4in}
        }
        \quad 
        \subcaptionbox{Yelp: Average Stars\label{fig:rev11YelpAvgStars}}{
                \vspace{-0.05in}
                \includegraphics[scale=0.26]{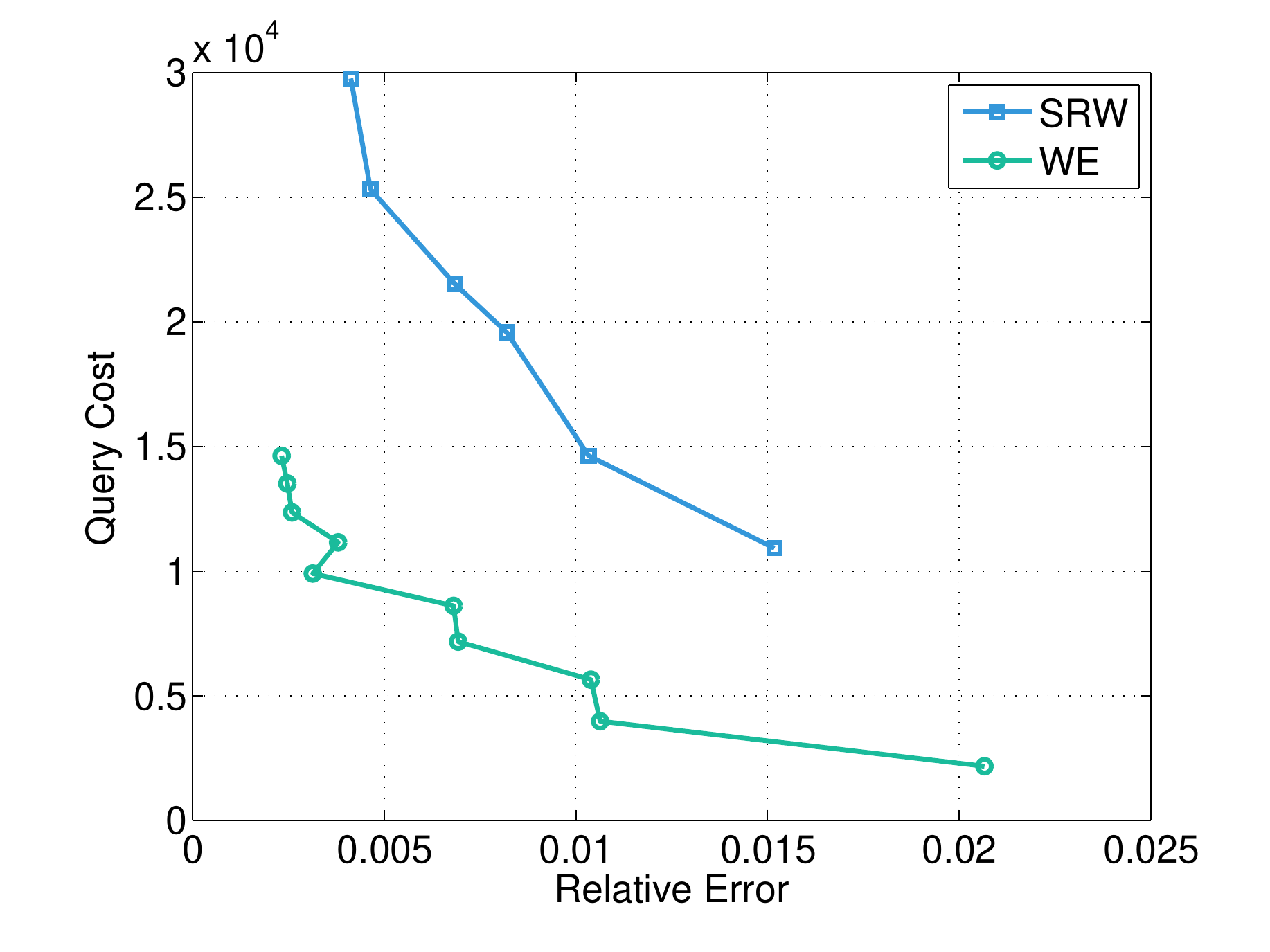}
                \vspace{-0.2in}
                \hspace{-0.4in}
        }
        %\caption{Relative Error of the Average Estimations vs Query Cost in Yelp.}
        %\label{fig:rev1Yelp1}
        \quad
        \subcaptionbox{Yelp: Average Shortest Path\label{fig:rev11YelpAvgShortestPath}}{
                \vspace{-0.05in}
                \includegraphics[scale=0.26]{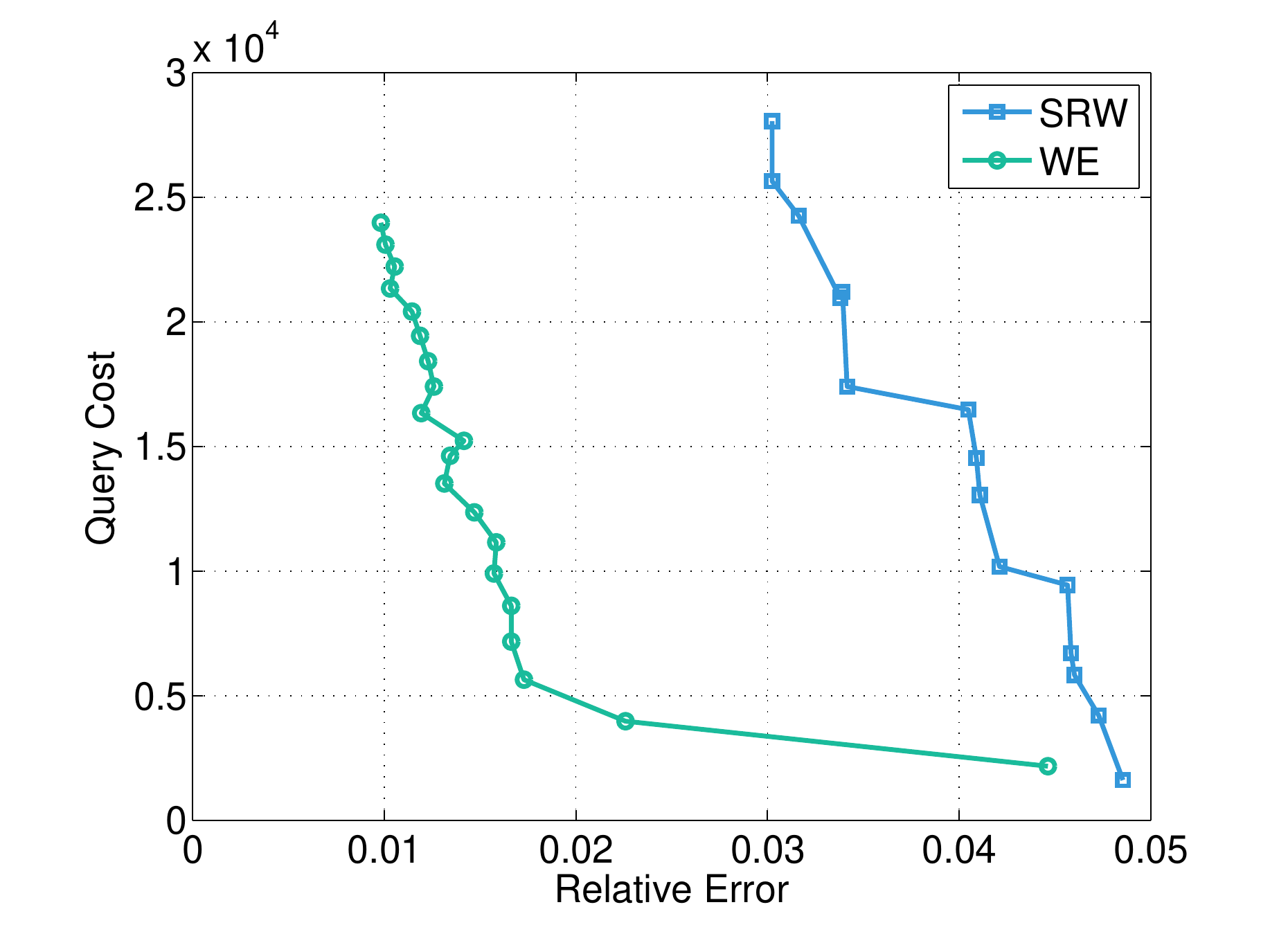}
                \vspace{-0.2in}
                \hspace{-0.4in}
        }
        \quad 
        \subcaptionbox{Yelp: Average Local Clustering Coefficient\label{fig:rev11YelpAvgLCC}}{
                \vspace{-0.05in}
                \includegraphics[scale=0.26]{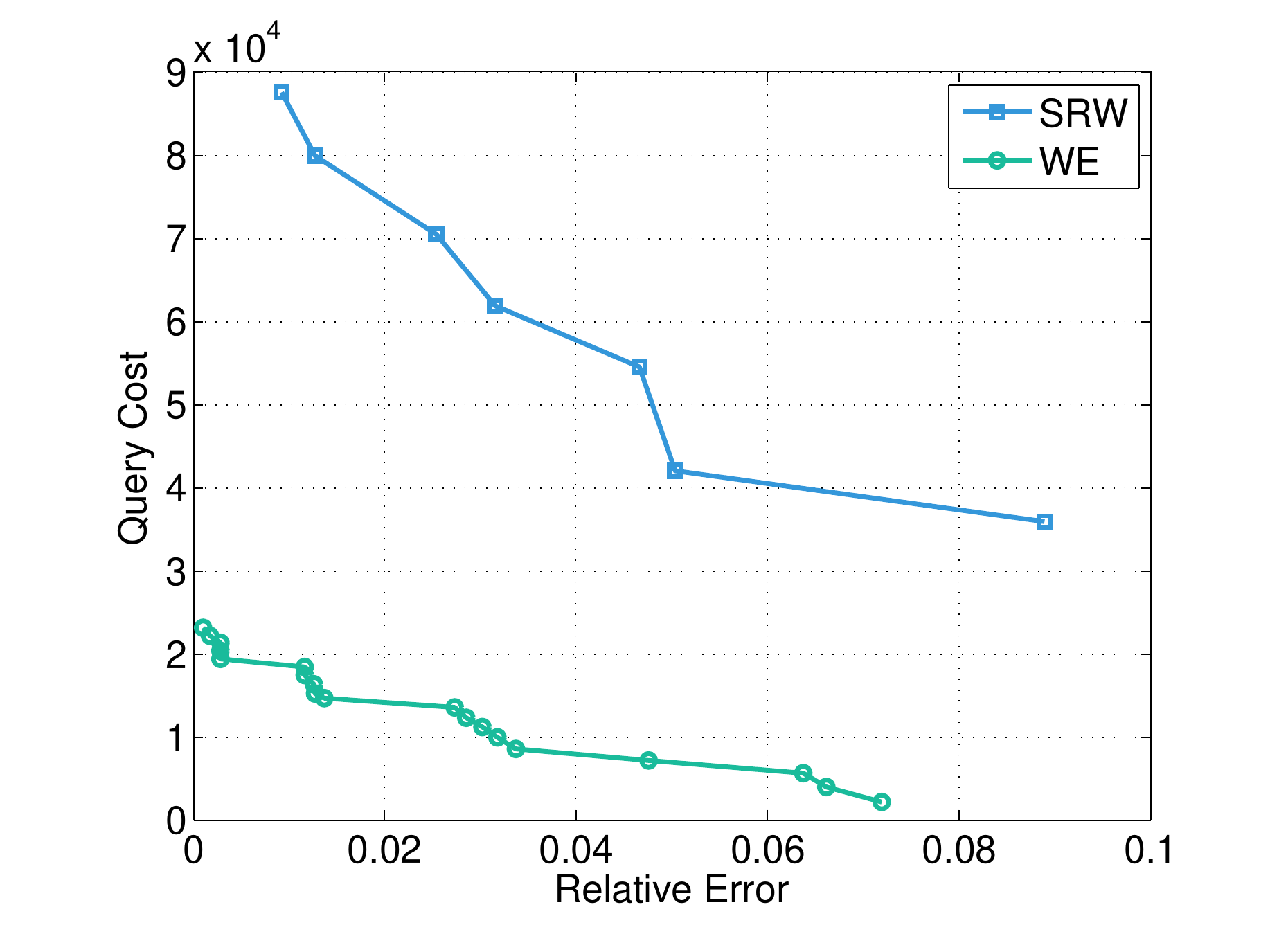}
                \vspace{-0.2in}
                \hspace{-0.4in}
        }
        \vspace{-0.1in}
        \caption{Relative Error of the Average Estimations vs Query Cost in Yelp.}
        \label{fig:Yelp}
\end{figure*}

\begin{figure*}[ht]
        \quad
        \subcaptionbox{Twitter: Average In-Degree \label{fig:rev12TwitterAvgInDeg}}{
                \vspace{-0.05in}
                \includegraphics[scale=0.26]{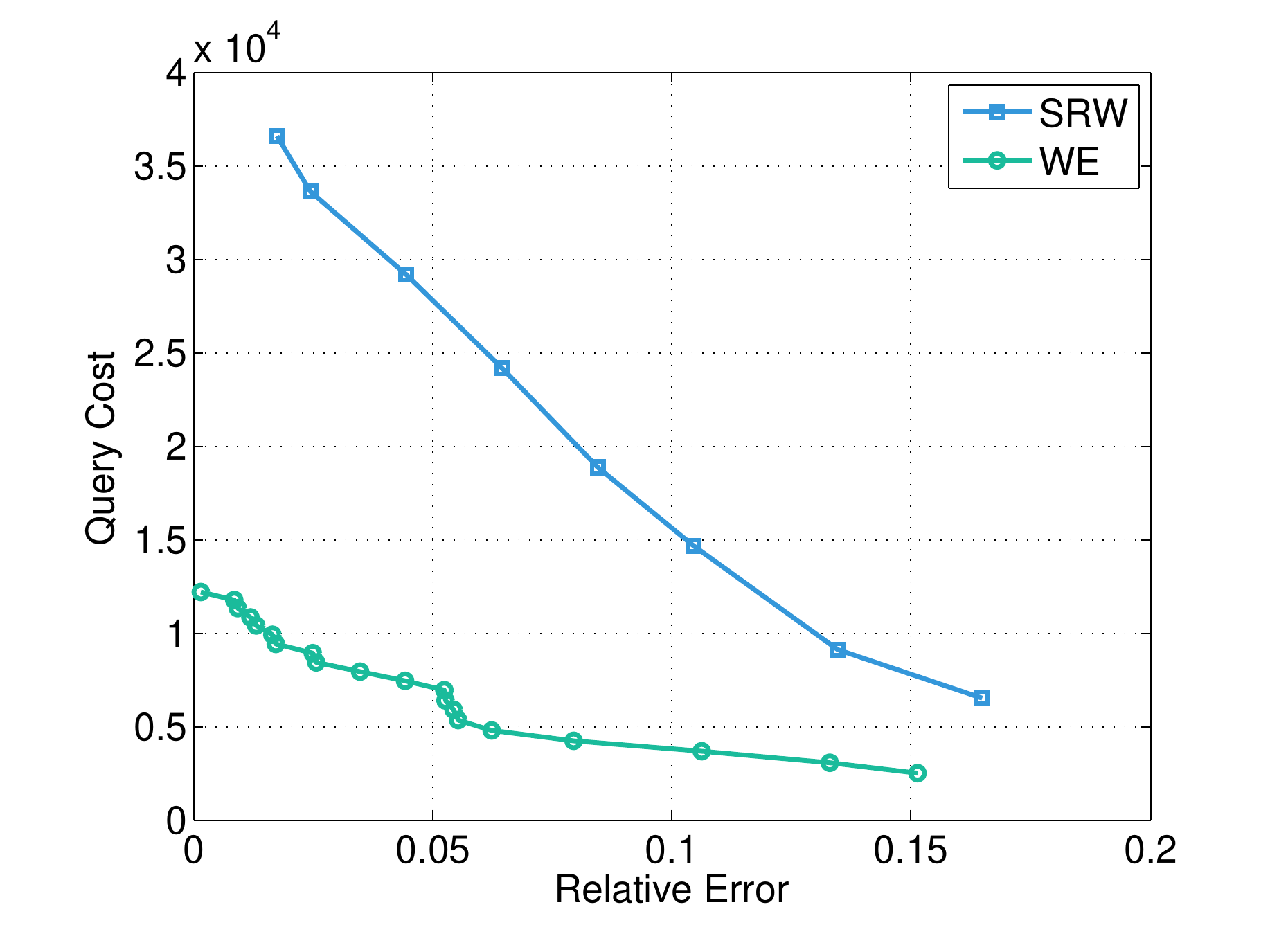}
                \vspace{-0.2in}
                \hspace{-0.4in}
        }
        \quad 
        \subcaptionbox{Twitter: Average Out-Degree\label{fig:rev12TwitterAvgOutDeg}}{
                \vspace{-0.05in}
                \includegraphics[scale=0.26]{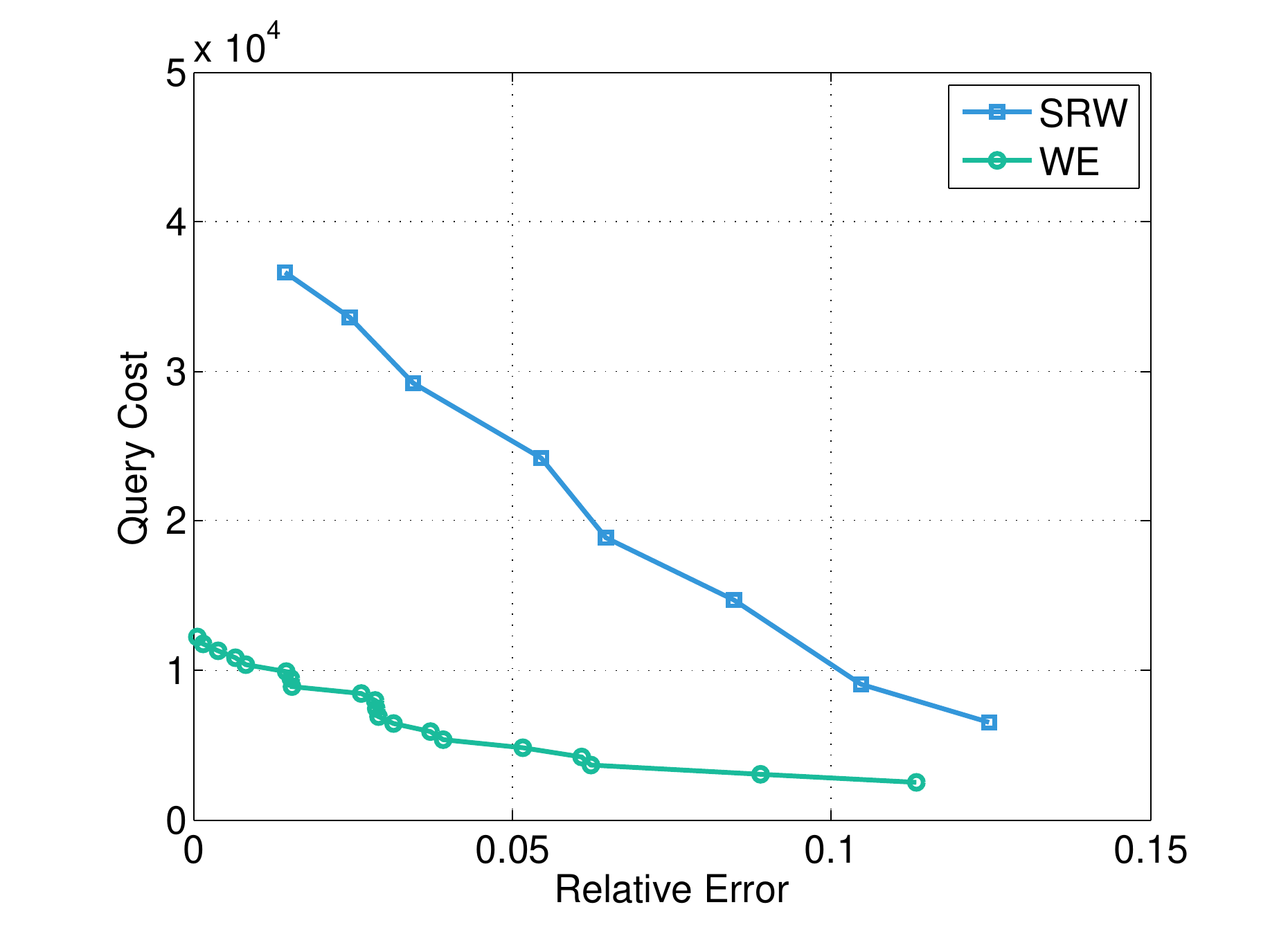}
                \vspace{-0.2in}
                \hspace{-0.4in}
        }
        \quad 
        \subcaptionbox{Twitter: Average Local Clustering Coefficient\label{fig:rev12TwitterAvgShortestPath}}{
                \vspace{-0.05in}
                \includegraphics[scale=0.26]{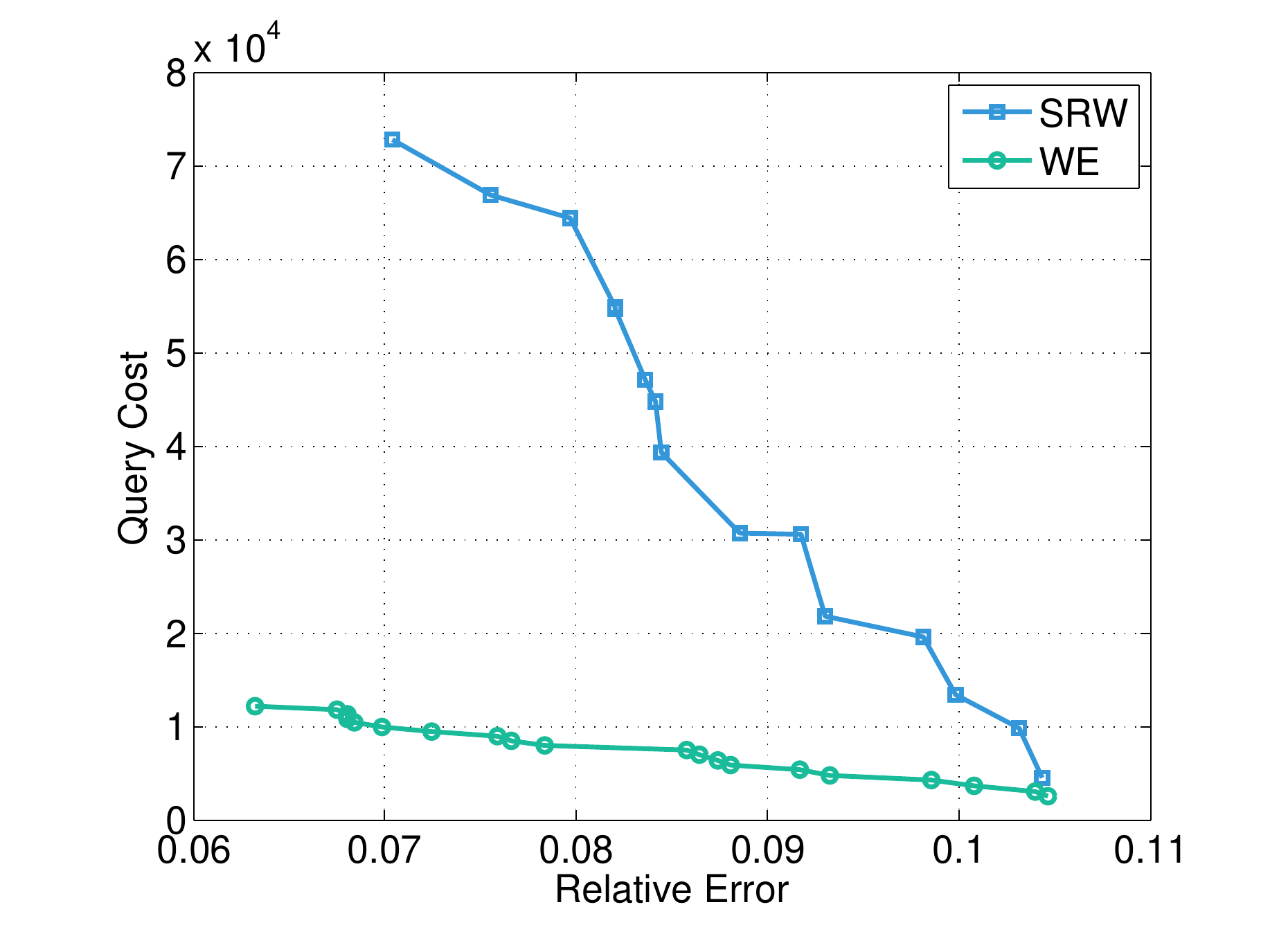}
                \vspace{-0.2in}
                \hspace{-0.4in}
                }
        \quad 
        \subcaptionbox{Twitter: Average Local Clustering Coefficient\label{fig:rev12TwitterAvgCC}}{
                \vspace{-0.05in}
                \includegraphics[scale=0.26]{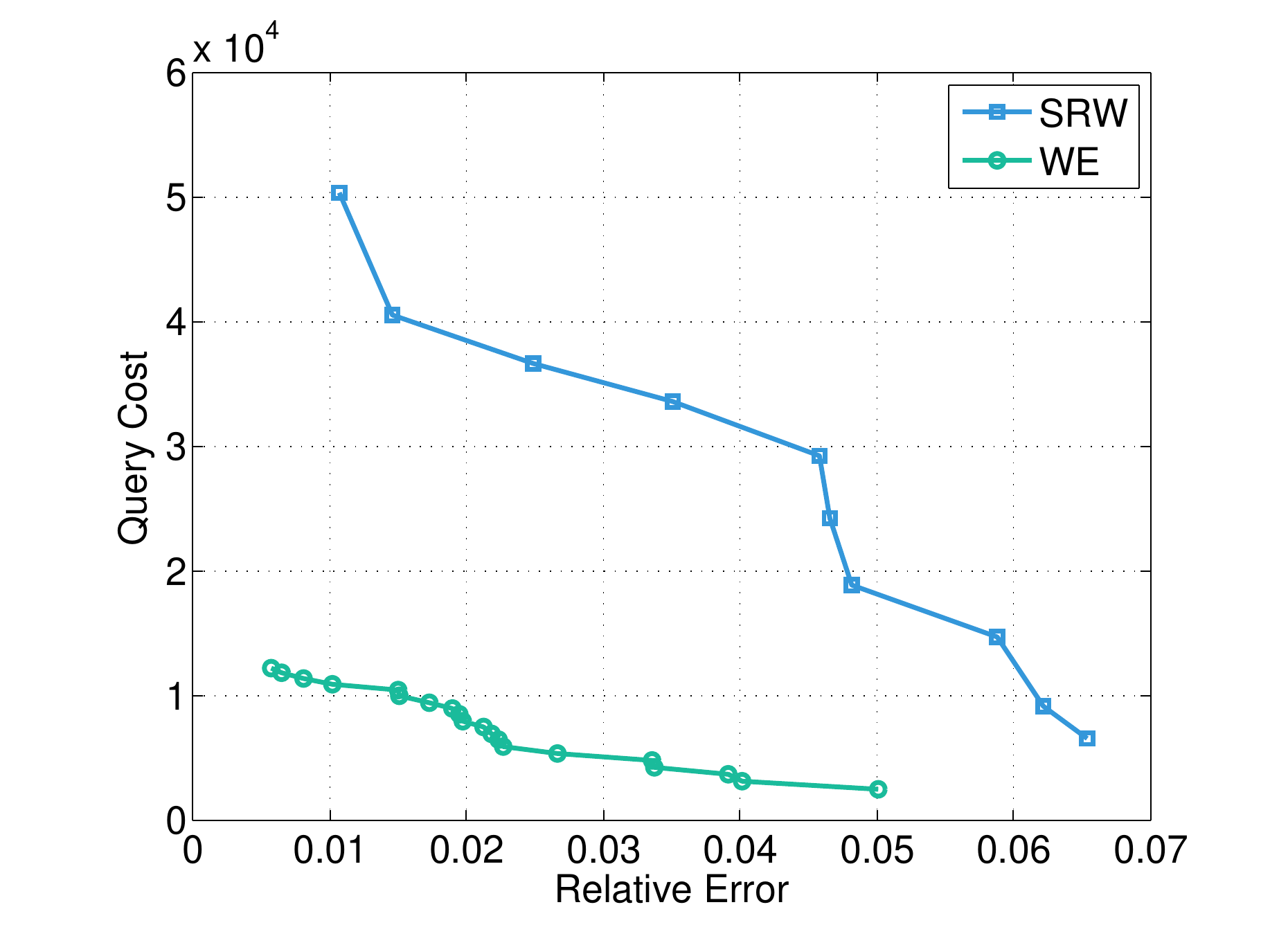}
                \vspace{-0.2in}
                \hspace{-0.4in}      
        }
        \caption{Relative Error of the Average Estimations vs Query Cost in Twitter (from SNAP repository).}
        \label{fig:Twitter1}
\end{figure*}

\begin{figure*}[ht]
	\centering
	\subcaptionbox{Average Degree (SRW)\label{fig:gplus_srw_relE_vs_QueryCost_trend}}{
		\vspace{-0.05in}
		\includegraphics[scale=0.26]{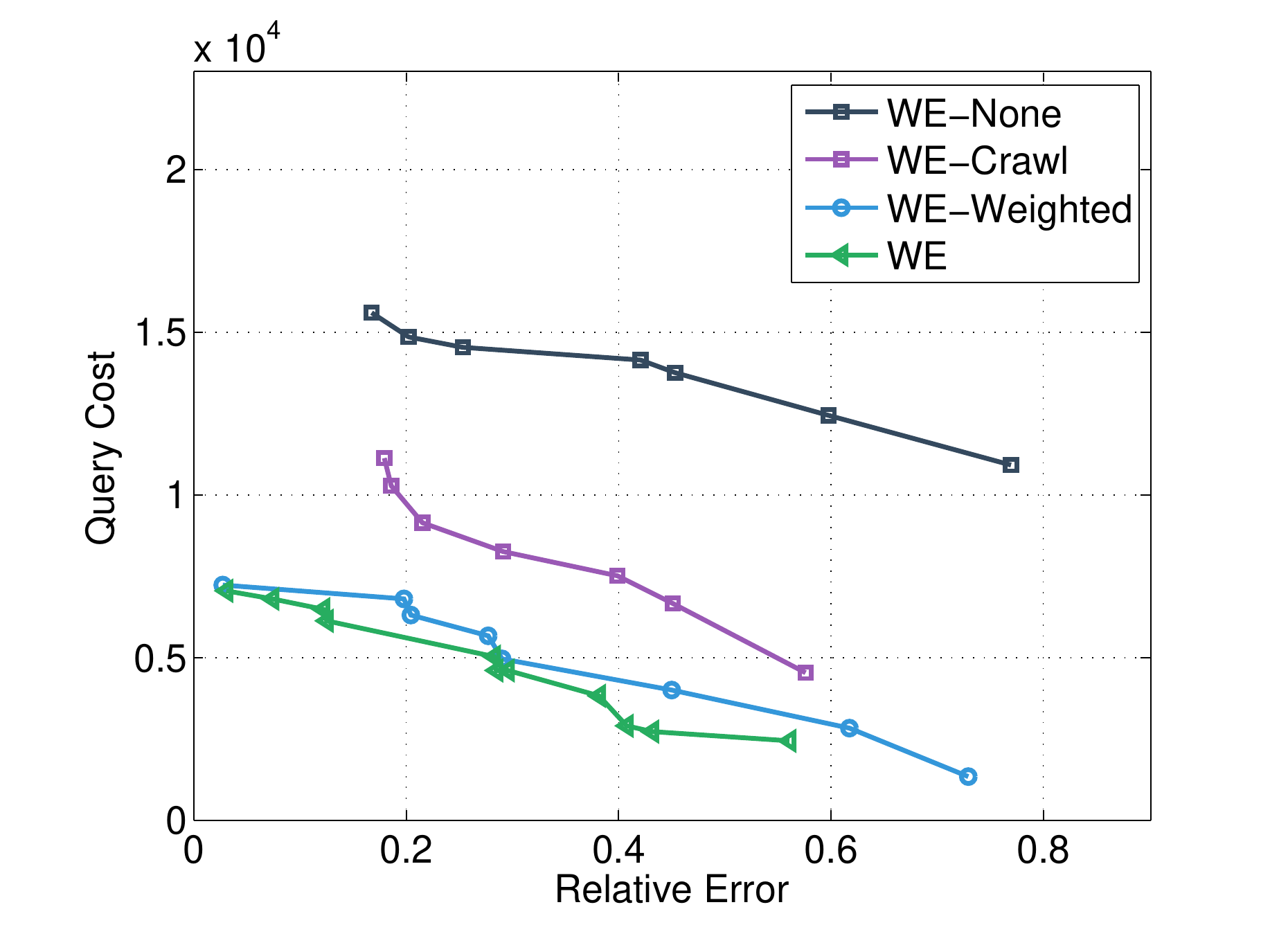}
		\vspace{-0.2in}
		\hspace{-0.4in}
	}
	\quad 
	\subcaptionbox{Average Self-description Length (SRW)\label{fig:gplus_srw_relE_vs_QueryCost_DL_trend}}{
	\vspace{-0.05in}
		\includegraphics[scale=0.26]{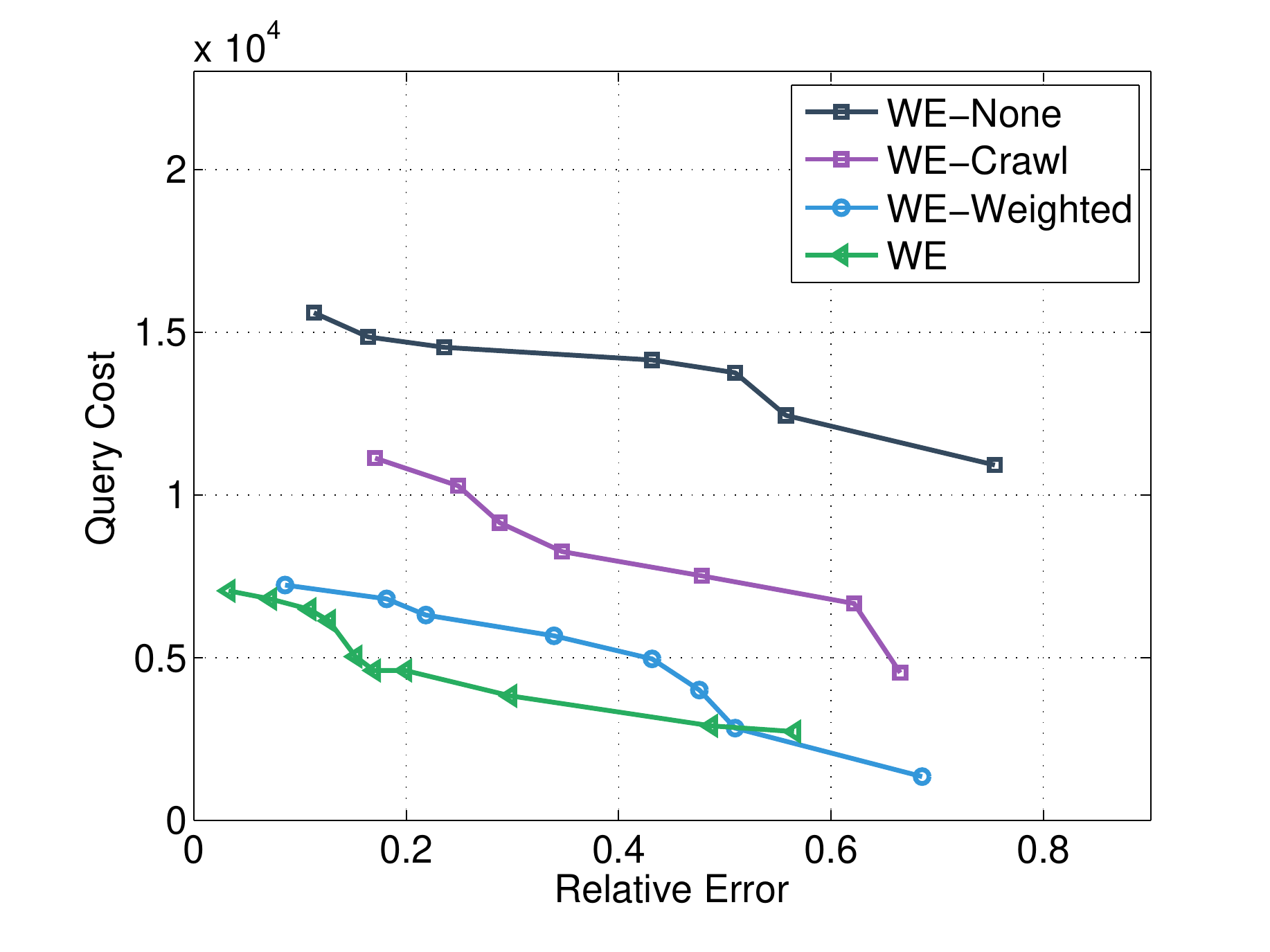}
		\vspace{-0.2in}
		\hspace{-0.4in}
	}
	\quad
	\subcaptionbox{Average Degree (MHRW)\label{fig:gplus_mhrw_relE_vs_QueryCost_trend}}{
	\vspace{-0.05in}
		\includegraphics[scale=0.26]{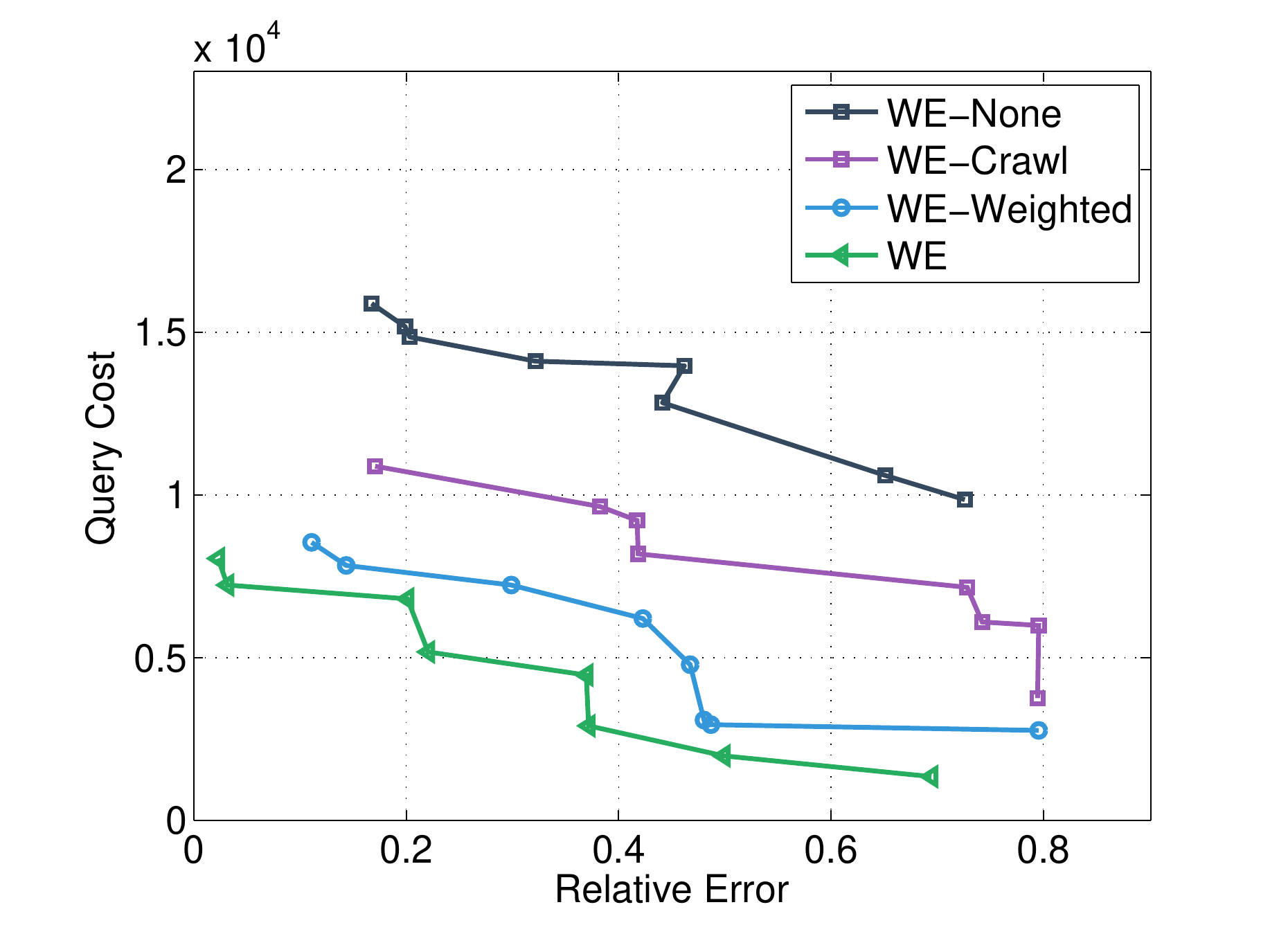}
		\vspace{-0.2in}
		\hspace{-0.4in}
	}
	\quad
	\subcaptionbox{Average Self-description Length (MHRW)\label{fig:gplus_mhrw_relE_vs_QueryCost_DL_trend}}{
	\vspace{-0.05in}
		\includegraphics[scale=0.26]{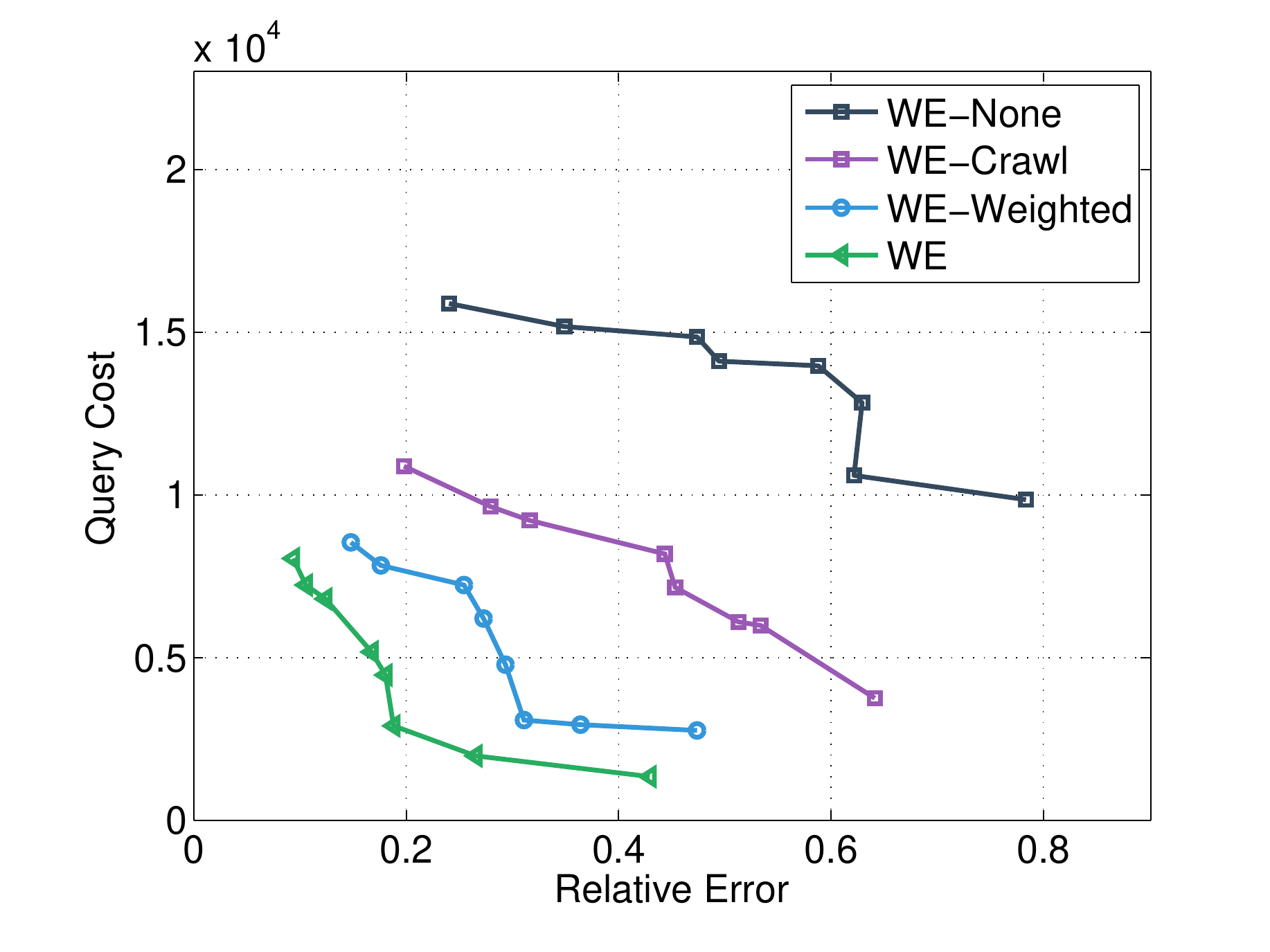}
		\vspace{-0.2in}
		\hspace{-0.4in}
	}
	\vspace{-0.1in}
	\caption{Improvement Trend: Relative Error of the Average Estimations vs Query Cost in Google Plus.}\label{fig:RelErrVSQC_trend}
	%\vspace{-0.2in}
\end{figure*}

\begin{figure*}[ht]
	\centering
	\subcaptionbox{Average Degree (SRW)\label{fig:gplus_srw_relE_vs_NumSamples}}{
		\vspace{-0.05in}
		\includegraphics[scale=0.26]{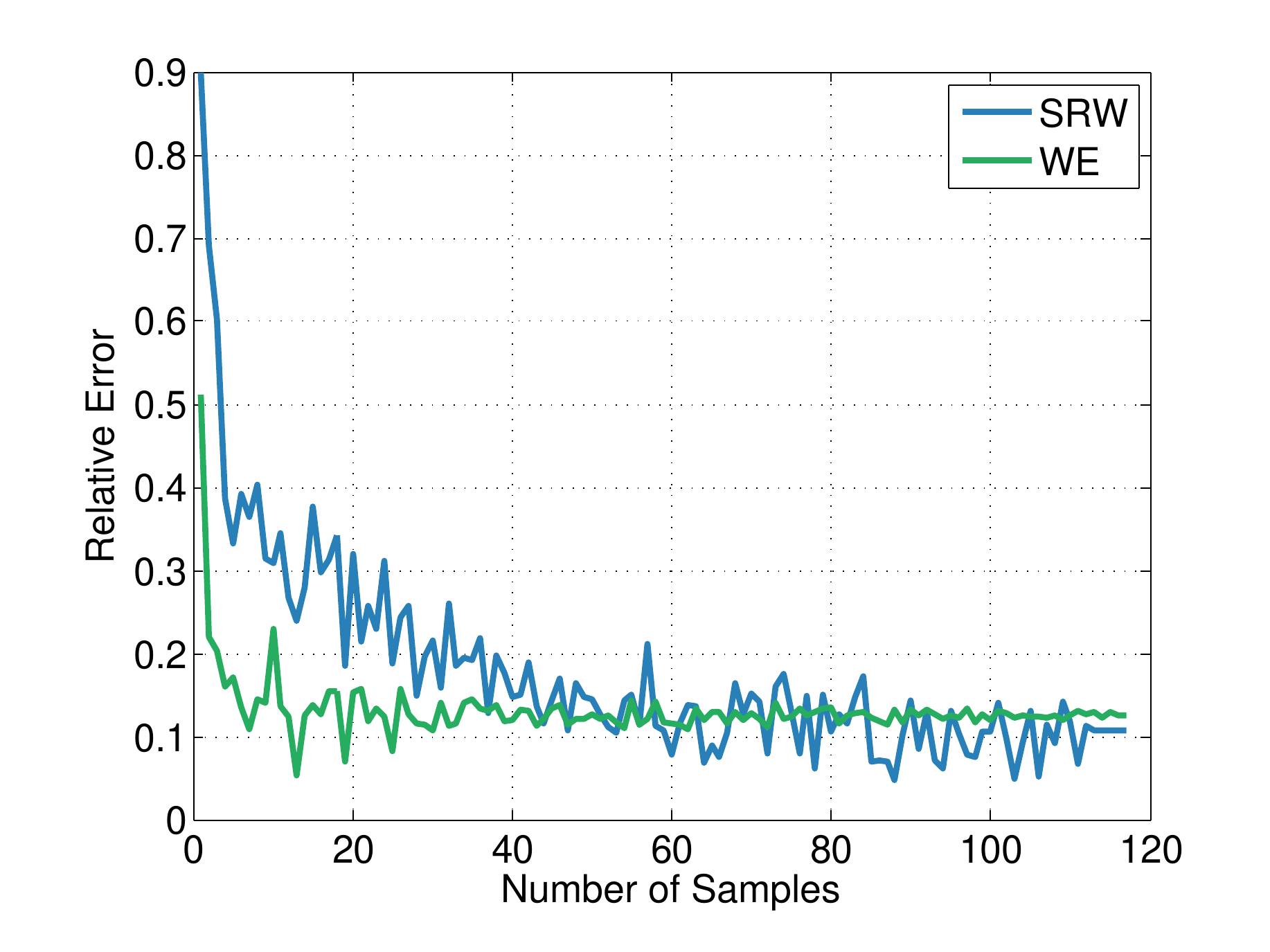}
		\vspace{-0.2in}
		\hspace{-0.4in}
	}
	\quad 
	\subcaptionbox{Average Self-description Length (SRW)\label{fig:gplus_srw_relE_vs_NumSamples_DL}}{
	\vspace{-0.05in}
		\includegraphics[scale=0.26]{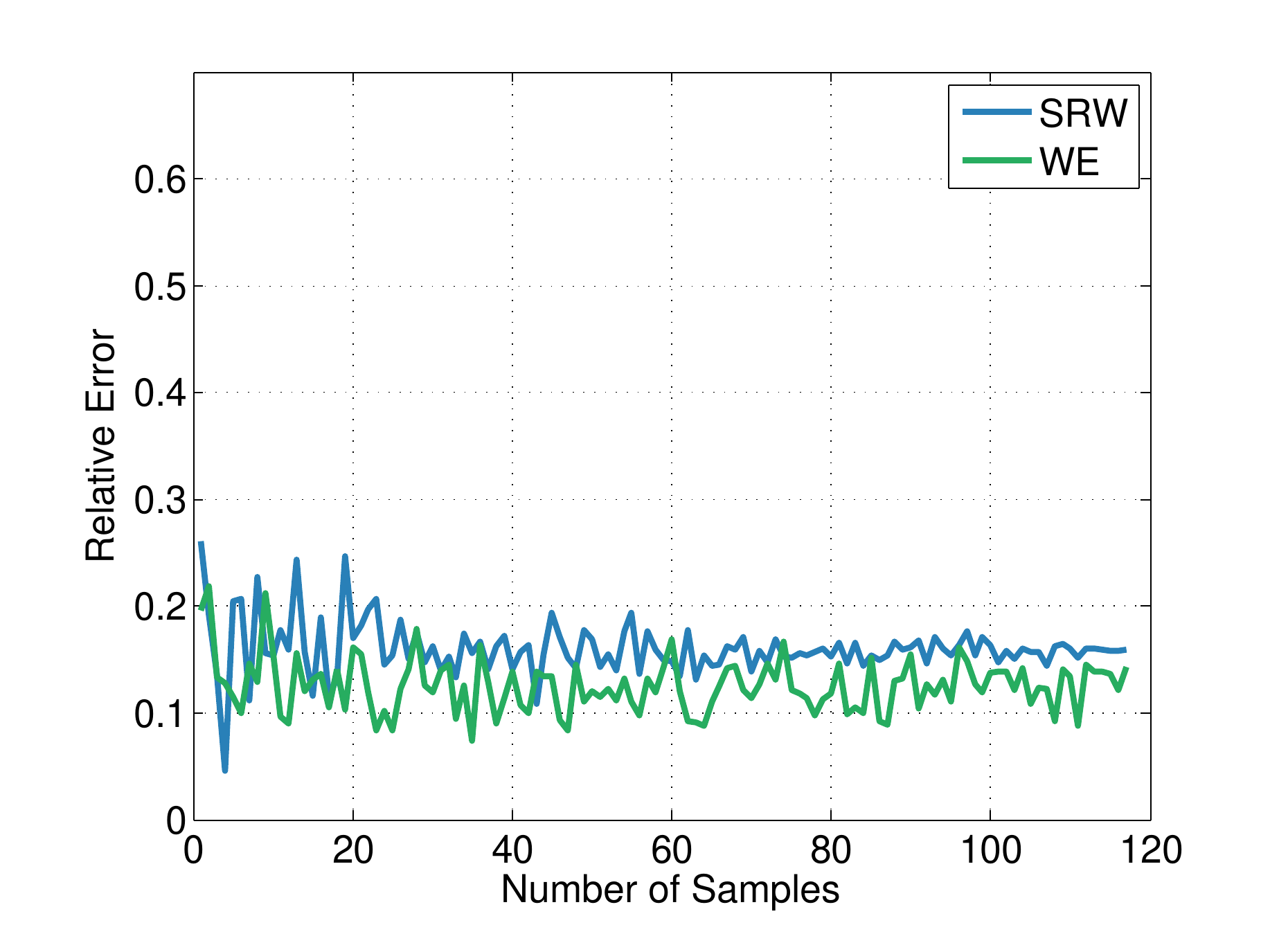}
		\vspace{-0.2in}
		\hspace{-0.4in}
	}
	\quad
	\subcaptionbox{Average Degree (MHRW)\label{fig:gplus_mhrw_relE_vs_NumSamples}}{
	\vspace{-0.05in}
		\includegraphics[scale=0.26]{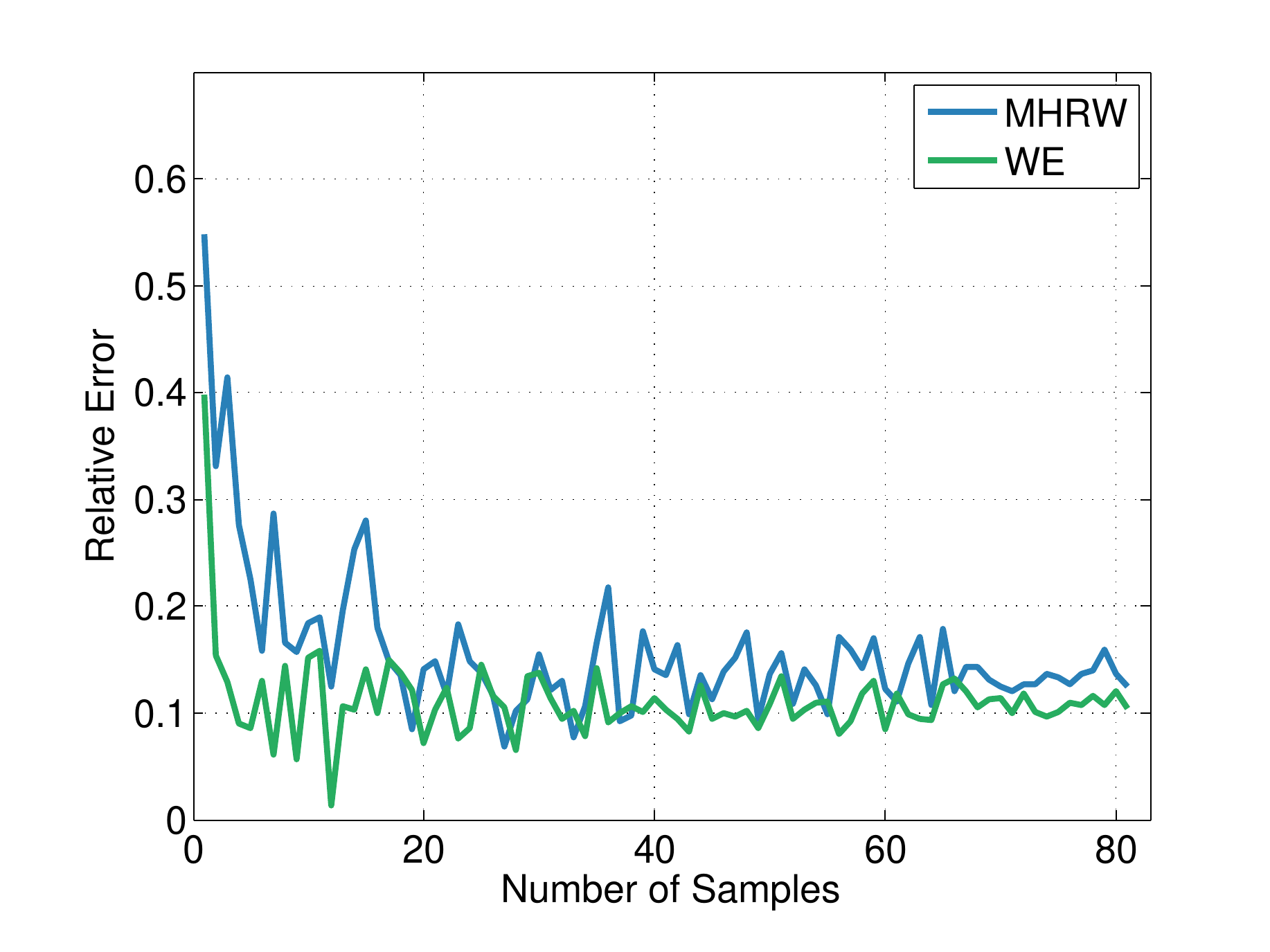}
		\vspace{-0.2in}
		\hspace{-0.4in}
	}
	\quad
	\subcaptionbox{Average Self-description Length (MHRW)\label{fig:gplus_mhrw_relE_vs_NumSamples_DL}}{
	\vspace{-0.05in}
		\includegraphics[scale=0.26]{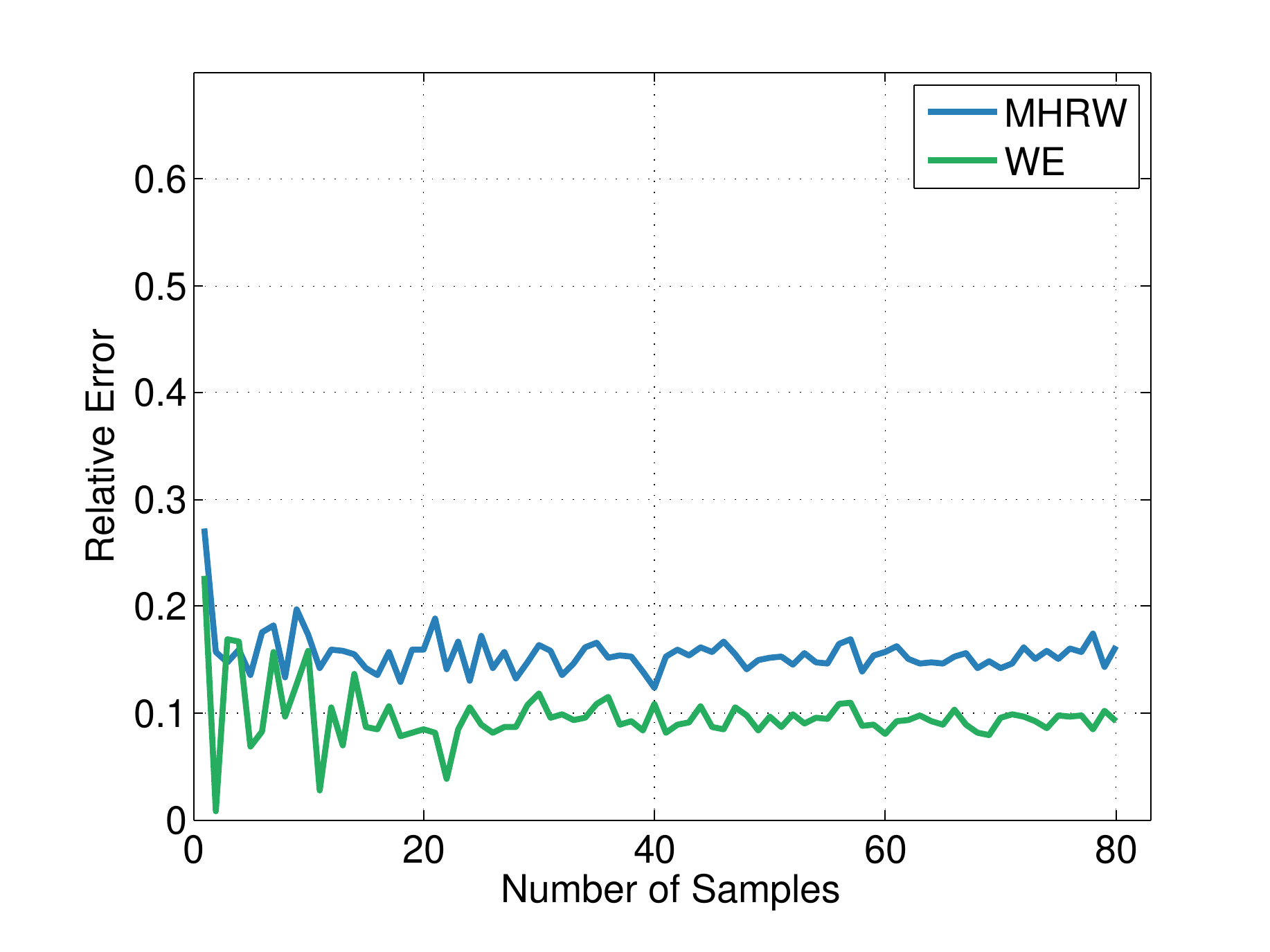}
		\vspace{-0.2in}
		\hspace{-0.4in}
	}
	\vspace{-0.1in}
	\caption{Relative Error of the Average Estimations vs Number of Samples in Google Plus.}\label{fig:RelErrVSNumSample}
	%\vspace{-0.2in}
\end{figure*}

\noindent {\bf Hardware and Platform:}
All our experiments were performed on a quad-core 2 GHz AMD Phenom machine running Ubuntu 14.04 with 8 GB of RAM. 
The algorithms were implemented in Python.

\vspace{1mm}
\noindent {\bf Datasets:} In this section, we test both synthetic and real-world data crawled from online social networks \madd{and also those which are publicly available}. Specifically, for synthetic data, we use the Barabasi-Albert model to generate scale-free networks with various sizes. For real-world data, \madd{we use three different popular social graphs, i.e. Google Plus, Yelp, and Twitter. The detail of each dataset is described bellow.} \del{we mainly use the Google Plus social graph described below. While we also conducted tests over the (largest connected component of) {\em Yelp Social Graph} published in Yelp's academic dataset (with 120,000 users and 954,000 edges), the results are similar to those of Google Plus and, due to space limitations, are not included in the paper.}

{\em Synthetic Graphs:} We generated scale-free networks with size ranging from 10,000 to 20,000 using the Barabasi-Albert model~\cite{Barabasi:1999} (implemented in Networkx~\cite{networkx:2008}), with the number of edges to attach from a new node (to existing nodes) set by default to 5.
%We also generated a scale-free network of size 1000 nodes and 6951 edges in order to compute an accurate sampling distribution.

{\em Google Plus Social Graph:} 
Google Plus\footnote{http://plus.google.com} is the second largest social networking site with more than 500 million active users.
For our experiments, we crawled a subset of the graph by starting from a number of popular users 
and recursively collecting information about their followers.
We model this dataset as an undirected graph where the users correspond to nodes and 
an edge exist between two users if at least one of them has the other in their circles.
We collected 16,405 users with more than 4.5 million connections between them.
The average degree of the graph is 560.44. We also collected each user's self description and used it in our tests.

\madd{
{\em Yelp Social Graph:} 
Yelp is a crowd-sourced local business review and social networking site with 132 million monthly visitors and 57 million reviews.
Yelp Academic dataset\footnote{\url{www.yelp.com/academic_dataset}} provides all the data and reviews of the 250 local businesses. For our experiments, we considered the largest connected component of user-user graph where nodes are the users and an edge exists between two users if they review atleast one similar business. Moreover, for each user there exists different information such as, review text, star rating, count of useful votes, count of funny votes, and count of cool votes. This graph has approximately 120,000 nodes and more than 954,000 edges.}

\madd{
{\em Twitter Social Graph:} 
Twitter is an online social network which is popular among millions of users who generate huge numbers of tweets, posts, and reviews every day. We used the Twitter dataset from Stanford's SNAP dataset repository\footnote{\url{snap.stanford.edu/data/egonets-Twitter.html}}
which is crawled from public sources and has close to 80,000 nodes and more than 1.7 million edges.}

%\madd{
%{\em Facebook Social Graph:} 
%We used a small benchmark of the Facebook dataset (facebook1684) from Stanford's SNAP dataset repository\footnote{\url{snap.stanford.edu/data/egonets-Facebook.html}}. It consists of $775$ nodes and $14006$ edges.}

\vspace{1mm}
\noindent {\bf Algorithms Evaluated:} 
We evaluated two traditional random walks - simple random walk (SRW) and Metropolis Hastings random walk (MHRW) - and the application of our WALK-ESTIMATE (WE) algorithm over each of them. Additionally, in order to evaluate the effect of the variance reduction heuristics, initial crawling and weighted sampling, proposed in Section\ref{sec:estimate}, we compared the performance of our main algorithm (WE) with three variations WE-None, WE-Crawl, and WE-Weighted. WE-None uses neither heuristics, WE-Crawl uses initial crawling only, while WE-Weighted uses weighted sampling only.

\vspace{1mm}
\noindent {\bf Parameter Settings:}
For SRW and MHRW, we used the Geweke convergence monitor~\cite{Cowles96markovchain} with threshold $Z \leq 0.1$. For our WALK component, we set the walk length to $2d+1$ where $d$ is the (estimated) graph diameter (set to $d = 7$ for Google Plus). For initial crawling, we set $h = 1$ for Google Plus and $h=2$ for the synthetic graphs\madd{, Yelp, and Twitter}. For weighted sampling, we set $\epsilon = 0.1$. \madd{We also considered 10th percentile of the estimation of sampling probabilities as the scale factor, $\min_{v \in V} (p(v)/q(v))$, for the acceptance rejection.}

\vspace{1mm}
\noindent {\bf Performance Measures:}
Given the large sizes of the graph being tested, it is impractical to precisely measure the bias of obtained samples. Thus, \madd{for the large graphs} we indirectly measured the sample bias by the relative error of AVG aggregate estimations generated from the samples (i.e., $|\tilde{x} - x|/x$ where $x$ and $\tilde{x}$ are the precise and estimated values of the aggregate, respectively). \madd{We used arithmetic and harmonic mean for the uniform and non-uniform  samples  respectively. We evaluate AVG aggregate of the measures related to the 
topological properties (such as degree, shortest path length, local clustering coefficient) as well as measures associated with a node (such as number of stars in Yelp, and user's self-description in Google Plus).} Specifically, for Google Plus, we considered two aggregates: the AVG degree and the AVG number of words in a user's self-description. \madd{ For Yelp, we estimated average number of stars and the topological properties average degree, average shortest path length and average local clustering coefficient.  For Twitter we estimated average in and out degrees (i.e. number of followers and followees), average shortest path length, and local clustering coefficient.} For the synthetic graphs, we only show the results for the the AVG degree. \del{was used because of the lack of other attributes. }
\madd {
Note that for each obtained datapoint in the results we reported average value of the 100 experiments.
Moreover, we computed exact bias of our algorithms. In order to compute an accurate sampling distribution, each node has to be sampled multiple times. This process is extremely time consuming and requires substantial query budget. Thus, computing the exact sampling distribution (and hence the bias) could only be done over small graphs. We used small scale-free network of size 1000 nodes and 6951 edges for this purpose. }
\subsection{Experimental Results}
\label{subsec:expResult}
\madd{\vspace{1mm}
\noindent {\bf Aggregate estimation:}}
We started by testing how WE performs against the baseline SRW and MHRW on the fundamental tradeoff in social network sampling - i.e., sample bias vs.~query cost. The results over Google Plus are shown in Figure 6. Specifically, subgraphs (a) and (b) depict SRW and WE with SRW as input random walk, while (c) and (d) are corresponding to MHRW. The AVG aggregate used to measure sample bias is AVG degree for (a) and (c) and AVG self-description length for (b) and (d). As one can see from the figure, our algorithm significantly outperforms SRW and MHRW - i.e., offers substantially smaller relative error for the same query cost - on both aggregates tested. \madd{Figure~\ref{fig:Yelp} shows the results over the Yelp dataset. Specifically, subgraph (a) shows the AVG aggregate of the node attributes, i.e. star rating while subgraphs (b), (c), and (d) is the results of  AVG aggregate of the topological properties, i.e,  degree, shortest path length, and clustering coefficient. The results confirm the fact that WE provides smaller relative error with the same query cost. We also test our algorithm in Twitter dataset and the results in Figure~\ref{fig:Twitter1} shows that AVG of the in-degree, out-degree, shortest path length, and clustering coefficient of the samples retrieved by the WE has smaller relative error than SRW for the same query cost.} 

We also study how our proposed variance reduction techniques improves the efficiency of our algorithm by comparing the performance of WE, WE-None, WE-Crawl and WE-Weighted, again according to how the relative error of aggregate estimation changes with the query cost. Figure~\ref{fig:RelErrVSQC_trend} depicts the result for Google Plus, according to the same subgraph setup (i.e., random-walk/aggregate combination) as Figure 6. One can see that, as expected in all cases, WE outperforms the single-heuristics variants, which in turn outperform the theoretical variant of the algorithm.

Next, we tested the quality of samples obtained by WE, in order to verify that the above-tested performance enhancements are not merely from walks being shorter, but from an equal or higher quality sample as well. To this end, Figure 8 depicts how the relative errors on AVG estimations change with the number of samples produced by SRW, MHRW, and the corresponding WEs, respectively on Google Plus - again according to the same subgraph setup as Figure 6. One can see that in all cases, the samples produced by WE achieves smaller relative error than the corresponding input random walks (with Geweke convergence monitor), indicating the smaller sample bias achieved by WE.

Finally, we performed the same tests, i.e., relative error vs.~query cost and vs.~sample size over the synthetic graphs with size varying from 10,000 to 20,000 - the results are depicted in Figure\ref{fig:syntheticRes}. Here we used SRW as the input random walk and AVG degree as the aggregate to be estimated. One can see from the figure that, while both WE and SRW requires a higher query cost for sampling over a larger graph, WE consistently outperforms SRW in all tested cases.

\madd{
\vspace{1mm}
\noindent {\bf Exact bias:}
 We computed exact bias of our algorithm over two small networks. Recall that sample bias is defined as the distance between actual sampling distribution (i.e., the probability distribution according to which each node is drawn as a sample)
and a predetermined target distribution such as uniform distribution. Multiple distance measures such as {\em variation distance}, K-L divergence can be used to quantify the bias. We measure the actual sampling probability for every node as follows.
We run the sampling algorithm with an extremely large query budget so that each node is sampled multiple times (eg 1000 times).
The sampling distribution is computed by counting the number of times each node is visited 
and is then compared with the target distribution to derive the bias. 
In order to compute an accurate sampling distribution, each node has to be sampled multiple times.
This process is extremely time consuming and requires substantial query budget.
Thus, computing the exact sampling distribution (and hence the bias) could only be done over small graphs. 

We used small scale-free network of size 1000 nodes and 6951 edges for this purpose. 
We compare three different sampling distributions:
(1) theoretical target distribution denoted as {\em theo}
(2) WE sampling distribution and 
(3) SRW sampling distribution. 
Table~\ref{tbl:rev2VariationalDistSynthGraph} provides the distance details and Figure~\ref{fig:rev2CompareBiasBA} shows the result of the experiment. The results confirm that the sampling distribution of WE is much closer than that of SRW.
\begin{table}[ht]
\small
    \centering
    \begin{tabular}{|c|c|c|}
        \hline
        {\bf Distance Measure} & {\bf Dist(Theoretical, SRW)}         & {\bf Dist(Theoretical, WE)} \\ \hline
        $\ell_{\infty}$     &  0.0081  & 0.00549 \\ \hline
        K-L Divergence      & 0.47529   & 0.01834 \\ \hline
    \end{tabular}
    \caption{\small{Distance between Theoretical Sampling Distribution and that of SRW/WE (Synthetic)}} 
    \label{tbl:rev2VariationalDistSynthGraph}
\end{table}
}

%% file: 8-related_work.tex
\section{Related Work}

\noindent {\bf Random Walks:} As discussed in Section 2, random walk is an MCMC based sampling method extensively studied in statistics (e.g, \cite{Gilks99}). Besides the traditional random walk designs described in Section 2, two key related concepts used in this paper are {\em burn-in period}, which captures the number of steps a random walk takes before converging to its stationary distribution \cite{Meyn:2009:MCS:1550713}; and convergence monitors, heuristic techniques for measuring on-the-fly how long the burn-in period should be (i.e., determining when a random walk should be stopped and a sample taken). Examples here include Geweke, Raftery and Lewis, Gelman and Rubin convergence monitors (see \cite{Cowles96markovchain} for a comprehensive review).

\vspace{1mm}
\noindent {\bf Random Walks on Social Networks:} There have been extensive studies (e.g., \cite{Leskovec2006a, Airoldi:2005:SAP:1117454.1117457, Kurant}) on the sampling of online social networks which feature graph browsing interfaces \cite{nan2011} that enforce the aforementioned local-neighborhood-only access limitation. \cite{Leskovec2006a} introduces a taxonomy of sampling techniques - specifically, node sampling, edge sampling and subgraph sampling. For the problem studied in this paper - i.e., sampling nodes from online social networks - the usage of multiple parallel random walks is studied in \cite{Alon2008}, while several studies (e.g., \cite{Leskovec2006a}) demonstrates the superiority of random walk techniques such as Simple Random Walk (SRW), Metropolis-Hastings Random Walk (MHRW) over baseline solutions such as Breadth First Search (BFS) and Depth First Search (DFS).  An interesting issue studied in the literature is the comparison between SRW and MHRW over real-world social networks - the finding in \cite{Gjoka2010} is that MHRW is less efficient than SRW because MHRW mixes more slowly. While our technique discussed in the paper is transparent to the input random walk, a similar comparison result can be observed from our experimental results as well.

\vspace{1mm}
\noindent {\bf Improving the Efficiency of Random Walks:} Most related to this paper are the previous studies on improving the efficiency of random walks over online social networks. To this end, \cite{Jin} combines random jump and MHRW to efficiently retrieve uniform random node samples from an online social networks. Nonetheless, in order to enable random jumps, this study assumes access to an {\em ID generator} which can sample a node uniformly at random with a high hit rate - an assumption that is not satisfied by many online social networks and not assumed in this paper. Another study \cite{Ribeiro:2010:ESG:1879141.1879192} considers frontier sampling which converts input samples with uniform distribution to output samples with arbitrary target distribution. Our study in the paper is transparent to this work - as we address the problem of generating sample nodes rather than assuming access to samples with pre-determined distributions. Also related to efficiency enhancements are \cite{Lee:2012:BRW:2254756.2254795} which introduces a non-backtracking random walk that converges faster with less asymptotic variance than SRW and \cite{Das:2013:FRW:2510649.2511334} which modifies the topology of the underlying graph on-the-fly in order to get a faster random walk on the modified graph. A key difference between WALK-ESTIMATE and all these existing studies is that while all existing techniques still wait for convergence to the target distribution, we do not wait for convergence, but rather proactively estimate the sampling distribution and then use rejection sampling to achieve the target distribution.

%% file: rw.bbl
\begin{thebibliography}{10}

\bibitem{Airoldi:2005:SAP:1117454.1117457}
E.~M. Airoldi.
\newblock Sampling algorithms for pure network topologies.
\newblock {\em SIGKDD Explorations}, 7:13--22, 2005.

\bibitem{almendral2007dynamical}
J.~A. Almendral and A.~D{\'\i}az-Guilera.
\newblock Dynamical and spectral properties of complex networks.
\newblock {\em New Journal of Physics}, 9(6):187, 2007.

\bibitem{Alon2008}
N.~Alon, C.~Avin, M.~Koucky, G.~Kozma, Z.~Lotker, and M.~R. Tuttle.
\newblock Many random walks are faster than one.
\newblock In {\em SPAA}, 2008.

\bibitem{backstrom2012four}
L.~Backstrom, P.~Boldi, M.~Rosa, J.~Ugander, and S.~Vigna.
\newblock Four degrees of separation.
\newblock In {\em Proc. of the 3rd Annual ACM Web Science Conf.}, pages 33--42.
  ACM, 2012.

\bibitem{Barabasi:1999}
A.-L. Barabasi and R.~Albert.
\newblock Emergence of scaling in random networks.
\newblock {\em Science}, 286(5439):509--512, 1999.

\bibitem{bollobas2004diameter}
B.~Bollob{\'a}s and O.~Riordan.
\newblock The diameter of a scale-free random graph.
\newblock {\em Combinatorica}, 24(1):5--34, 2004.

\bibitem{bonneau2009prying}
J.~Bonneau, J.~Anderson, and G.~Danezis.
\newblock Prying data out of a social network.
\newblock In {\em Social Network Analysis and Mining, 2009. ASONAM'09.
  International Conference on Advances in}, pages 249--254. IEEE, 2009.

\bibitem{catanese2011crawling}
S.~A. Catanese, P.~De~Meo, E.~Ferrara, G.~Fiumara, and A.~Provetti.
\newblock Crawling facebook for social network analysis purposes.
\newblock In {\em Proceedings of the international conference on web
  intelligence, mining and semantics}, page~52. ACM, 2011.

\bibitem{chau2007parallel}
D.~H. Chau, S.~Pandit, S.~Wang, and C.~Faloutsos.
\newblock Parallel crawling for online social networks.
\newblock In {\em Proceedings of the 16th international conference on World
  Wide Web}, pages 1283--1284. ACM, 2007.

\bibitem{Cohen:2003}
R.~{Cohen} and S.~{Havlin}.
\newblock {Scale-Free Networks Are Ultrasmall}.
\newblock {\em Phys. Rev. Lett.}, 90, 2003.

\bibitem{Cowles96markovchain}
M.~K. Cowles and B.~P. Carlin.
\newblock Markov chain monte carlo convergence diagnostics: A comparative
  review.
\newblock {\em Journal of the American Statistical Association}, 91:883--904,
  1996.

\bibitem{Dasgupta:2007}
A.~Dasgupta, G.~Das, and H.~Mannila.
\newblock A random walk approach to sampling hidden databases.
\newblock In {\em SIGMOD}, 2007.

\bibitem{Geyer1992}
C.~J. Geyer.
\newblock Practical markov chain monte carlo.
\newblock {\em Statistical Science}, 1992.

\bibitem{Gilks99}
W.~R. Gilks.
\newblock {\em Markov Chain Monte Carlo In Practice}.
\newblock Chapman and Hall/CRC, 1999.

\bibitem{Gjoka2010}
M.~Gjoka, M.~Kurant, C.~T. Butts, and A.~Markopoulou.
\newblock Walking in facebook: A case study of unbiased sampling of osns.
\newblock In {\em INFOCOM}, 2010.

\bibitem{networkx:2008}
A.~A. Hagberg, D.~A. Schult, and P.~J. Swart.
\newblock Exploring network structure, dynamics, and function using {NetworkX}.
\newblock In {\em Proc. of the 7th Python in Science Conf.}, Aug. 2008.

\bibitem{Jerrum:1988:CRM:62212.62234}
M.~Jerrum and A.~Sinclair.
\newblock Conductance and the rapid mixing property for markov chains: The
  approximation of permanent resolved.
\newblock In {\em STOC}, 1988.

\bibitem{Jin}
L.~Jin, Y.~Chen, P.~Hui, C.~Ding, T.~Wang, A.~V. Vasilakos, B.~Deng, and X.~Li.
\newblock Albatross sampling: robust and effective hybrid vertex sampling for
  social graphs.
\newblock In {\em MobiArch}, 2011.

\bibitem{301511}
R.~E. Kass, B.~P. Carlin, A.~Gelman, and R.~M. Neal.
\newblock {Markov Chain Monte Carlo in Practice: A Roundtable Discussion}.
\newblock {\em American Statistician}, 52:93--100, 1998.

\bibitem{Katzir2011}
L.~Katzir, E.~Liberty, and O.~Somekh.
\newblock {Estimating sizes of social networks via biased sampling}.
\newblock In {\em WWW}, 2011.

\bibitem{Kleinberg:2000}
J.~Kleinberg.
\newblock The small-world phenomenon: An algorithmic perspective.
\newblock In {\em Proc. of the Thirty-second Annual ACM Symposium on Theory of
  Computing}. ACM, 2000.

\bibitem{Kurant}
M.~Kurant, M.~Gjoka, C.~T. Butts, and A.~Markopoulou.
\newblock Walking on a graph with a magnifying glass: stratified sampling via
  weighted random walks.
\newblock In {\em SIGMETRICS}, 2011.

\bibitem{kwak2010twitter}
H.~Kwak, C.~Lee, H.~Park, and S.~Moon.
\newblock What is twitter, a social network or a news media?
\newblock In {\em Proc. of the 19th int. conf. on World wide web}, pages
  591--600. ACM, 2010.

\bibitem{Lee:2012:BRW:2254756.2254795}
C.-H. Lee, X.~Xu, and D.~Y. Eun.
\newblock Beyond random walk and metropolis-hastings samplers: Why you should
  not backtrack for unbiased graph sampling.
\newblock SIGMETRICS, 2012.

\bibitem{Leskovec2006a}
J.~Leskovec and C.~Faloutsos.
\newblock Sampling from large graphs.
\newblock In {\em SIGKDD}, 2006.

\bibitem{Levin:2008}
D.~A. Levin, Y.~Peres, and E.~L. Wilmer.
\newblock {\em {Markov Chains and Mixing Times}}.
\newblock American Mathematical Society, 2008.

\bibitem{Liben-nowell05analgorithmic}
D.~Liben-nowell and E.~D. Demaine.
\newblock An algorithmic approach to social networks.
\newblock Technical report, PhD thesis at MIT References 118 Science and
  Artificial Intelligence Laboratory, 2005.

\bibitem{Lovasz1993}
L.~Lov\'{a}sz.
\newblock {Random walks on graphs: A survey}.
\newblock {\em Combinatorics, Paul Erdos is Eighty}, 2(1):1--46, 1993.

\bibitem{Meyn:2009:MCS:1550713}
S.~Meyn and R.~L. Tweedie.
\newblock {\em Markov Chains and Stochastic Stability}.
\newblock Cambridge University Press, 2nd edition, 2009.

\bibitem{mislove2007measurement}
A.~Mislove, M.~Marcon, K.~P. Gummadi, P.~Druschel, and B.~Bhattacharjee.
\newblock Measurement and analysis of online social networks.
\newblock In {\em SIGCOMM}, 2007.

\bibitem{Mohaisena}
A.~Mohaisen, A.~Yun, and Y.~Kim.
\newblock {Measuring the mixing time of social graphs}.
\newblock In {\em SIGCOMM}, 2010.

\bibitem{nazir2008unveiling}
A.~Nazir, S.~Raza, and C.-N. Chuah.
\newblock Unveiling facebook: a measurement study of social network based
  applications.
\newblock In {\em Proceedings of the 8th ACM SIGCOMM conference on Internet
  measurement}, pages 43--56. ACM, 2008.

\bibitem{Ribeiro:2010:ESG:1879141.1879192}
B.~Ribeiro and D.~Towsley.
\newblock Estimating and sampling graphs with multidimensional random walks.
\newblock In {\em SIGCOMM}, 2010.

\bibitem{robson1964sample}
D.~Robson and H.~Regier.
\newblock Sample size in petersen mark--recapture experiments.
\newblock {\em Transactions of the American Fisheries Society}, 93(3):215--226,
  1964.

\bibitem{seary2003spectral}
A.~Seary and W.~Richards.
\newblock Spectral methods for analyzing and visualizing networks: an
  introduction.
\newblock {\em Dynamic Social Network Modeling and Analysis}, pages 209--228,
  2003.

\bibitem{Shyu:2013}
E.~Shyu.
\newblock Diameter bounds and eigenvalues.
\newblock Technical report, MIT, 2013.

\bibitem{tsiatas2013spectral}
A.~Tsiatas, I.~Saniee, O.~Narayan, and M.~Andrews.
\newblock Spectral analysis of communication networks using dirichlet
  eigenvalues.
\newblock In {\em Proceedings of the 22nd international conference on World
  Wide Web}, pages 1297--1306. International World Wide Web Conferences
  Steering Committee, 2013.

\bibitem{vilnis2013markov}
L.~Vilnis.
\newblock Markov chain monte carlo, mixing, and the spectral gap.
\newblock 2013.

\bibitem{wondracek2010practical}
G.~Wondracek, T.~Holz, E.~Kirda, and C.~Kruegel.
\newblock A practical attack to de-anonymize social network users.
\newblock In {\em Security and Privacy (SP), 2010 IEEE Symposium on}, pages
  223--238. IEEE, 2010.

\bibitem{ye2010crawling}
S.~Ye, J.~Lang, and F.~Wu.
\newblock Crawling online social graphs.
\newblock In {\em Web Conference (APWEB), 2010 12th International
  Asia-Pacific}, pages 236--242. IEEE, 2010.

\bibitem{nan2011}
N.~Zhang and G.~Das.
\newblock Exploration of deep web repositories.
\newblock In {\em Proc. of the VLDB Endowment (VLDB), Tutorial}, 2011.

\bibitem{Das:2013:FRW:2510649.2511334}
Z.~Zhou, N.~Zhang, Z.~Gong, and G.~Das.
\newblock Faster random walks by rewiring online social networks on-the-fly.
\newblock ICDE, 2013.

\bibitem{Zhu:2013}
Z.~A. Zhu, S.~Lattanzi, and V.~S. Mirrokni.
\newblock A local algorithm for finding well-connected clusters.
\newblock {\em CoRR}, 2013.

\end{thebibliography}
